\documentclass[twocolumn]{autart}    

\usepackage{cite}
\usepackage{amsmath,amssymb,amsfonts}
\usepackage{algorithm}
\usepackage{algorithmicx}
\usepackage{algpseudocode}
\usepackage{graphicx}
\usepackage{mathrsfs}
\usepackage{textcomp}
\usepackage{xcolor}
\usepackage{caption}
\usepackage{multirow} 
\usepackage{array,booktabs}
\usepackage{booktabs, tabularx, dcolumn, array}
\newcolumntype{d}[1]{D{.}{.}{#1}}
\newcolumntype{C}{>{\centering\arraybackslash}X}

\newtheorem{theorem}{Theorem}
\newtheorem{definition}{Definition} 
\newtheorem{lemma}{Lemma} 
\newtheorem{remark}{Remark}

\newtheorem{proposition}{Proposition}
\newtheorem{assumption}{Assumption}

\begin{document}

\begin{frontmatter}

\title{Optimality Robustness in Koopman-Based Control \thanksref{footnoteinfo}} 

\thanks[footnoteinfo]{This paper was not presented at any IFAC 
	meeting. Corresponding author Z.~Duan.}

\author[SAMR]{Yicheng Lin}\ead{linyc020709@stu.pku.edu.cn},
\author[SAMR]{Bingxian Wu}\ead{davidwu2003@stu.pku.edu.cn},
\author[HK]{Nan Bai}\ead{eenanbai@ust.hk},
\author[COE]{Yunxiao Ren}\ead{renyx@pku.edu.cn},
\author[SAMR]{Zhongkui Li}\ead{zhongkli@pku.edu.cn},
\author[SAMR]{Zhisheng Duan}\ead{duanzs@pku.edu.cn}

\address[SAMR]{School of Advanced Manufacturing and Robotics, Peking University, Beijing}
\address[HK]{Department of Electronic and Computer Engineering, The Hong Kong University of Science and Technology, Hong Kong}
\address[COE]{College of Engineering, Peking University, Beijing}
          
\begin{keyword}                           
Data-driven control theory, Koopman operator, nonlinear systems, robustness analysis, robust optimal control.
\end{keyword}                             

\begin{abstract}                          
The Koopman operator enables simplified representations for nonlinear systems in data-driven optimal control, but the accompanying uncertainties inevitably induce deviations in the optimal controller and associated value function. This naturally raises the question of how such uncertainty-induced optimality deviation can be quantified and mitigated.
To address this problem, we adopt a unified analysis-to-design perspective that connects the characterization of optimality robustness with its improvement through controller design.
At the analysis level, we establish a unified treatment of multiple uncertainty sources in Koopman-based control, where approximation error and noisy data are incorporated into a common robustness analysis through a norm-bounded representation.
At the design level, we develop a robustness-aware optimal control methodology that provably reduces such optimality deviations, thereby enhancing robustness while explicitly revealing a quantitative trade-off between nominal optimality and robustness.
As for practical implementation aspect, we further propose a tractable policy iteration algorithm, whose well-posedness and convergence are established via vanishing viscosity regularization and elliptic partial differential equation (PDE) techniques. Numerical examples validate the theoretical findings and demonstrate the effectiveness of proposed methodology.
\end{abstract}

\end{frontmatter}

\section{Introduction}
The Koopman operator has become a powerful framework for simplifying nonlinear system representations \cite{koopman1931hamiltonian}. Particularly, the linear Koopman operator is able to transform finite-dimensional nonlinear autonomous systems into infinite-dimensional linear ones \cite{2020Optimal}, and transform control-affine nonlinear systems into bilinear forms \cite{goswami2021bilinearization}. This has stimulated researchers to apply the well-established control theory of linear or bilinear systems to investigate nonlinear system control problems, e.g., stability analysis \cite{mauroy2013spectral,mauroy2016global}, feedback stabilization \cite{CDC2022You,lin2025integratinguncertainties} and optimization \cite{villanueva2021towards,TAC2025KoopmanHJB}. During the past decade, the outstanding potential of Koopman operator in data-driven control settings has been increasingly explored \cite{TSMC2024Resilient},\cite{YvTRO2024Autogeneration}.

In data-driven control practice via the Koopman operator, researchers typically use the well-known (extended) dynamic mode decomposition (DMD/EDMD) algorithm \cite{2017On} to approximate the Koopman operator in a finite-dimensional function space, or identify lifted models of nonlinear systems with chosen dictionary functions. These models then serve as the basis for controller synthesis, for example, using model predictive control (MPC) \cite{KORDA2018KMPC,9867811} or learning-based methods \cite{CALDARELLI2025LQRKN, 9867541}. However, due to intrinsic limitations of finite-dimensional approximation and imperfect data acquisition, various uncertainties are unavoidable in this process. Specifically, uncertainties can be categorized by sources \cite{lin2025integratinguncertainties,Strasser2024Koopman}.
\begin{itemize}
	\item \textit{Projection error arising from finite dictionary functions}, caused by truncating the infinite-dimensional system representation to a finite-dimensional one.
	\item \textit{Estimation error due to finite data collection}, arising from system identification with a finite dataset.
	\item \textit{Noise or disturbance in data collection}, where measured states are corrupted by measurement noise or external disturbances.
\end{itemize}
From the perspective of uncertainty sources, the first two categories originate from intrinsic limitations of the modeling and identification process and are collectively referred to as \textit{approximation error}, whereas the third category arises from exogenous factors in data acquisition. All these uncertainties are inherent to Koopman-based data-driven control and have a direct impact on the achieved control performance.

Recent years have witnessed growing efforts toward reducing the approximation error while learning the Koopman operator \cite{IACOB202249} and identifying the lifted systems \cite{singh2026deeprobustkoopmanlearning}. Meanwhile, several studies on Koopman-based control have investigated the impacts of different uncertainties, but primarily focusing on robust stability. For instance, probabilistic and deterministic bounds on the approximation error have been established in \cite{Strasser2024Koopman,strasser2025kernelbasederrorboundsbilinear} and subsequently used for feedback stabilization. In \cite{lin2025integratinguncertainties}, the impacts of approximation error and process disturbance were unified to ensure closed-loop robust stability. Robust Koopman-based MPC (RK-MPC) approaches were developed in \cite{9867811,ZHANG2022RKMPC} to address constraint satisfaction and robust stability under the impacts of approximation error and additive bounded noise in data collection.

Despite these advances, existing studies have largely confined robustness analysis to stability-oriented objectives, which represent only basic requirements in controller design. When optimal controllers are synthesized based on identified models without explicitly accounting for the approximation error and noise in data collection, the resulting control laws effectively solve a surrogate problem rather than the original nonlinear optimal control task. Consequently, the notion of optimality in such settings may become misleading, as there may exist substantial optimality deviations between the obtained policies and the true optimal solutions. Quantitative characterization and mitigation of the optimality deviations are therefore essential for reliable performance in data-driven optimal control.

Although classical linear-quadratic and $H_2/H_\infty$ control theories provide valuable insights into robustness and sensitivity \cite{aliyu2011nonlinear}, most notably through Riccati-based perturbation and performance analyses, these results are largely restricted to linear systems with specific structural assumptions. For general nonlinear and data-driven settings, particularly in Koopman-based control, a systematic theoretical framework for characterizing the optimality deviations remains largely unexplored.

In this work, we investigate optimality robustness in Koopman-based data-driven control by examining how uncertainties affect the value function and associated control policy through their deviations from the actual optimal solution. This optimality-oriented viewpoint complements conventional robustness analysis and robust optimal control approaches, which have primarily addressed stability preservation or constraint satisfaction. Adopting an analysis-to-design perspective, this paper integrates optimality robustness analysis with robustness-aware controller synthesis. Accordingly, the main contributions can be summarized along two main lines.
\begin{itemize}
	\item \textit{Robustness analysis}: We introduce a unified robustness analysis framework for Koopman-based data-driven control under multiple uncertainty sources. A key technical contribution is showing that energy-bounded noise in data collection can be transformed into a norm-bounded perturbation (\textbf{Theorem~\ref{Thm: Noise to Error}}), enabling noisy data and approximation error to be jointly analyzed. The resulting optimality deviations are quantitatively characterized (Propositions~\ref{Prop: performance deviation} \& \ref{Prop: controller deviation}), providing the basis for robustness-aware controller synthesis.
	
	\item \textit{Controller design}: Building on the robustness analysis, we formulate a robust optimal control methodology that explicitly accounts for worst-case uncertainties and provably reduces optimality deviations (\textbf{Theorem~\ref{Thm: Robust Controller}}). Further, we establish a fundamental trade-off between nominal optimality and robustness (\textbf{Theorem~\ref{Thm: Robust Performance}}), revealing how an acceptable loss of nominal performance yields a guaranteed increase in robustness. For practical implementation, we develop a policy iteration algorithm, whose convergence (\textbf{Theorem~\ref{Thm: Convergence}}) is established via vanishing viscosity regularization and elliptic PDE theory that overcome the inapplicability of existing results \cite{JZP2022TNNLSContraction,LEE2021PIRLDecrease}.
\end{itemize}

The optimality robustness analysis in this work is closely connected to the general theoretical results developed in our related work \cite{lin2025optimalitydeviationusingkoopman}. Different from that work, which focuses on the fundamental robustness analysis in nonlinear optimal control, the present manuscript investigates how such robustness insights can be incorporated into Koopman-based data-driven control, including uncertainty handling, controller synthesis, and computational realization.

The remainder of this paper is organized as follows. Section~\ref{2.preliminaries} reviews preliminaries on the Koopman operator and lifted system representations. Section~\ref{3.approximation error} analyzes the optimality deviation caused by the approximation error. Section~\ref{4.noise} investigates the impact of energy-bounded noise in data collection. Section~\ref{5.robust design} presents the robust optimal controller design methodology for correcting the optimality deviations, covering both theoretical formulation and practical algorithm. Section~\ref{6.numerical} validates the analysis results and proposed design methodology with numerical examples. Section~\ref{7.conclusion} concludes this paper.

\textbf{Notations:} Throughout this paper, we denote by $\mathbb{R}^n$ the $n$-dimensional Euclidean space. The norm for real vector $v\in\mathbb{R}^{n_v}$ is Euclidean norm $\|v\|=\sqrt{\sum_{i=1}^{n_v}v_i^2} =\sqrt{v^\top v}$, and the norm for real matrix $M=(m_{ij})\in \mathbb{R}^{n_r\times n_c}$ is Frobenius norm $\|M\|=\sqrt{\sum_{i=1}^{n_r} \sum_{j=1}^{n_c}m_{ij}^2}=\sqrt{\text{trace}(M^\top M)}$. The null space and column space of $M$ are denoted by $\mathrm{Nul}(M)$ and $\mathrm{Col}(M)$ respectively, while the direct sum of subspaces is denoted by $\oplus$. For a symmetric $M$, we denote by $\lambda_{\min}(M)$ and $\lambda_{\max}(M)$ its minimum and maximum eigenvalue. For two symmetric matrices $M_1,M_2$, relation $M_1 \succeq M_2$ ($M_1 \preceq M_2$) means that matrix $M_1-M_2$ is positive (negative) semidefinite.

We denote by $\mathcal{C}^0(\cdot)$ the space of continuous functions defined on the corresponding domains, $\mathcal{C}^k(\cdot)$ the space of functions with continuous (partial) derivatives up to order $k$, and $\mathcal{C}^{k,\alpha}(\cdot)$ the subspace of $\mathcal{C}^k(\cdot)$ consisting of functions whose $k$-th order (partial) derivatives are uniformly H\"older continuous \cite{gilbarg2001elliptic} with exponent $\alpha\in[0,1]$. For a measurable real-valued function $f:\Omega\rightarrow\mathbb{R}$, $L^p$-norm is $\|f\|_{L^p(\Omega)}=\left(\int_{\Omega}|f(x)|^p\right)^{1/p}$, $1\leq p\leq\infty$, and the space of functions satisfying $\|f\|_{L^p(\Omega)}<\infty$ is denoted by $L^p(\Omega)$. In particular, the space of essentially bounded measurable functions on $\Omega$ is denoted by $L^\infty(\Omega)$, where $L^\infty$-norm is $\|f\|_{L^\infty(\Omega)}=\text{ess}\sup_{x\in\Omega}|f(x)|$.

The big-O notation is primarily used to characterize how quantities depend on some key parameters (e.g., error bound coefficients), i.e., the relation $X=\mathcal{O}(Y)$ indicates that there exists $C_{XY}>0$ such that $X\leq C_{XY}Y$ holds.

\section{Preliminaries}\label{2.preliminaries}
Consider the unactuated nonlinear system
\begin{equation}\label{unactuated system}
	\dot{x}(t)=f\left(x(t)\right)
\end{equation}
defined on a state space $\mathbb{X}\subseteq\mathbb{R}^n$ and $f(0)=0$, i.e., the origin is an equilibrium of unactuated system. We define a Banach space $\mathcal{F}\subseteq \mathcal{C}^0(\mathbb{X})$ of observables $\varphi:\mathbb{X}\rightarrow\mathbb{R}$, and the Koopman operator is defined as follows \cite{mauroy2020koopman}.

\begin{definition}
	The continuous time Koopman operator $\mathcal{K}^t:\mathcal{F} \rightarrow \mathcal{F}$ is defined as
	\begin{equation}
		(\mathcal{K}^t\varphi)(x_0)=\varphi\circ S(t,x_0) \label{operator definition}
	\end{equation}
	where $\circ$ denotes the function composition and $S(t,x_0)$ denotes the flow map (solution) of system \eqref{unactuated system} at time $t>0$ with initial state $x(0)=x_0$. Furthermore, assuming $\varphi(x)\in\mathcal{C}^1(\mathbb{X})$ is continuously differentiable, it satisfies
	\begin{equation}
		\frac{\mathrm{d}\varphi(x)}{\mathrm{d}t}=\mathcal{L}_f\varphi\triangleq\lim_{t\to0}\frac{(\mathcal{K}^t\varphi-\varphi)}{t}=\nabla\varphi\cdot f \label{generator definition}
	\end{equation}
	where $\mathcal{L}_f$ is defined as the infinitesimal generator that equals to the Lie derivative with respect to $f$.
\end{definition}

Notably the Koopman operator is linear even if the system dynamics is nonlinear since for any $\alpha,\beta\in\mathbb{R}$ and $\varphi_1,\varphi_2\in\mathcal{F}$, $\mathcal{K}^t(\alpha \varphi_1+\beta \varphi_2)=\alpha \mathcal{K}^t\varphi_1 +\beta \mathcal{K}^t\varphi_2$. This linearity naturally paves the way to its spectral property, characterized by the Koopman eigenvalues and eigenfunctions \cite{mauroy2020koopman,kvalheim2021existence}.

\begin{definition}
	A Koopman eigenfunction corresponding to the unactuated system \eqref{unactuated system} is an observable $\phi_\lambda\in\mathcal{F}$ such that
	\begin{equation}
		\mathcal{K}^t\phi_\lambda=e^{\lambda t}\phi_\lambda
	\end{equation}
	for some $\lambda\in\mathbb{C}$, which is the associated Koopman eigenvalue. Additionally with \eqref{generator definition}, we obtain the following property of eigenfunctions
	\begin{equation}
		\mathcal{L}_f\phi_\lambda=\nabla\phi_\lambda\cdot f=\lambda\phi_\lambda. \label{eigenfunction definition generator}
	\end{equation}
\end{definition}
It is known that if $\phi_{\lambda_1},\phi_{\lambda_2}$ are Koopman eigenfunctions with eigenvalues $\lambda_1,\lambda_2$ respectively, $\phi_{\lambda_1}^{k_1}\phi_{\lambda_2}^{k_2}$ is also an eigenfunction with eigenvalue $k_1\lambda_1+k_2\lambda_2$, which means that there are perhaps infinitely many eigenfunctions.

Koopman eigenvalues and eigenfunctions play important roles in nonlinear system representation. Define a set of dictionary function $\Psi(x)=[\psi_1(x)\ \ldots \ \psi_N(x)]^\top\in\mathbb{R}^N$ serving as the transformation function. In order to completely capture nonlinear dynamics, $N>n$ is the usual case, so the system after transformation is often called lifted system. If a set of Koopman eigenfunctions $\Phi(x)=[\phi_{\lambda_1}(x)\ \ldots \ \phi_{\lambda_N}(x)]^\top$ is chosen to be the dictionary, \eqref{unactuated system} is transformed into a totally linear system
\begin{equation*}
	\frac{\mathrm{d}}{\mathrm{d}t}\Phi(x)=\Lambda\Phi(x)
\end{equation*}
where $\Lambda=\text{diag}(\lambda_1,\ldots,\lambda_N)$. It is simple to prove that if the selected dictionary $\Psi(x)$ is equivalent to $\Phi(x)$, i.e., $\exists T\in\mathbb{R}^{N\times N}, \text{s.t.}\Psi(x)=T\Phi(x)$ with matrix $T$ being full rank, the lifted system is also linear. Besides relying on Koopman eigenfunctions, dictionary functions can be selected via different ways, for example, adopting various polynomials \cite{TSMC2024Resilient}, or more complicated kernels \cite{Strasser2024Koopman}. Here we make some standard assumptions, which are widely-adopted in the Koopman operator researches \cite{KORDA2018KMPC,ZHANG2022RKMPC}.

\begin{assumption}
	The dictionary functions defined on a compact, forward-invariant state space $\mathbb{X}$ satisfy \\
	(a) $\{\psi_i\}_{i=1}^N$ are continuously differentiable on $\mathbb{X}$, and naturally $\Psi(x)$ is Lipschitz continuous on $\mathbb{X}$ with Lipschitz constant $L_p$. \\
	(b) $\Psi(0)=0$, $\Psi^{-1}(0)\cap\mathbb{X}=\{0\}$. \\
	(c) $\exists C\in\mathbb{R}^{n\times N}, \text{s.t. }x=C\Psi(x)=Cz$.
\end{assumption}

Now we consider the control-affine nonlinear system
\begin{equation}\label{original system}
	\dot{x}(t)=f\left(x(t)\right)+\sum_{i=1}^{m}g_i\left(x(t)\right)u_i(t)
\end{equation}
where $x\in \mathbb{R}^n$ and $u=\left[u_1\ \ldots\ u_m\right]^\top\in\mathbb{R}^m$ are the state and control input respectively, $f,g_i\in \mathcal{C}^1(\mathbb{R}^n),i=\{1,\ldots,m\}$. Strictly defining the Koopman operator for nonlinear systems with input is a nontrivial task, which is discussed in a recent work \cite{haseli2025roadskoopmanoperatortheory} in detail. But undoubtedly, the nonlinear system \eqref{original system} can be simplified via the Koopman operator using the chosen dictionary function $\Psi(x)$, whose dynamics satisfies
\begin{equation}\label{lifted system operator}
	\frac{\mathrm{d}\Psi}{\mathrm{d}t}=\frac{\partial\Psi}{\partial x}f(x)+\sum_{i=1}^mu_i \frac{\partial\Psi}{\partial x}g_i(x)=\mathcal{L}_f\Psi+\sum_{i=1}^mu_i\mathcal{L}_{g_i}\Psi.
\end{equation}
Note that although lifted linear time invariant (LTI) system representations have demonstrated empirical effectiveness in many applications, they are theoretically valid only for restricted classes of nonlinear control systems \cite{haseli2025roadskoopmanoperatortheory}. However, lifting a control-affine nonlinear system to a bilinear form has been shown feasible in \cite{goswami2021bilinearization} if the eigenspace of $\mathcal{L}_f$ is an invariant subspace of $\mathcal{L}_{g_i},i=\{1,\ldots,m\}$. If this is not satisfied, a projection error term can be introduced to guarantee the equivalence with \eqref{lifted system operator}, which can be sufficiently small as shown in \cite{goswami2021bilinearization}. Therefore, to preserve the generality of proposed results in this paper, we considered the lifted bilinear form to characterize the dynamics of original nonlinear system \eqref{original system}, i.e.,
\begin{equation}\label{lifted system AB}
	\frac{\mathrm{d}z}{\mathrm{d}t}=Az+B_0u+\sum_{i=1}^m u_iB_iz+r(z,u).
\end{equation}
In data-driven control settings, matrices $A,B_0$ and $B_i,i=\{1,\ldots,m\}$ are identified from data (see Section~\ref{4.noise}-A for identification method), and not only projection error but also estimation error resulting from finite data are contained in $r(z,u)$. The proportional bound of approximation error $r(z,u)$ given by the following assumption possesses a certain degree of generality.
\begin{assumption}\label{Assum: error bound}
	Suppose there exist constants $c_1,c_2>0$ such that the approximation error term in \eqref{lifted system AB}, including projection error and estimation error, is bounded by
	\begin{equation}\label{error bound}
		\|r(z,u)\|\leq c_1\|z\|+c_2\|u\|.
	\end{equation}
	Furthermore, the partial derivative of approximation error $\frac{\partial r}{\partial u}$ is assumed to exist and be continuous.
\end{assumption}

The error bound has been investigated in several recent papers which indicate that Assumption~\ref{Assum: error bound} is not strong also the truth in a large number of cases. For instance, a probabilistic bound for the estimation error $c_1,c_2\in \mathcal{O}(1/\sqrt{\delta d_0})$ was derived by \cite[Proposition 5]{Strasser2024Koopman}, where $\delta,d_0$ denote the probability tolerance and data amount respectively. Leveraging kernel-based methods, a deterministic bound for the approximation error was established by \cite[Theorem 5]{strasser2025kernelbasederrorboundsbilinear}. These results ensure the universality of proportional error bound relationship \eqref{error bound}, also indicate that $c_1,c_2$ are relatively small (even sufficiently small in \cite{goswami2021bilinearization}). In fact, the approximation error term can be represented by
\begin{equation}
	r(z,u)=\mathcal{L}_f\Psi-A\Psi+\sum_{i=1}^mu_i(\mathcal{L}_{g_i}\Psi-B_{0,i}-B_i\Psi)
\end{equation}
where $B_{0,i}$ denotes the $i$-th column of $B_0$, then $\frac{\partial r}{\partial u}=\sum_{i=1}^m(\mathcal{L}_{g_i}\Psi-B_{0,i}-B_i\Psi)e_i$ where $e_i$ denotes the unit vector with the $i$-th component $1$. This illustrates the continuous differentiability of $r(z,u)$ with respect to $u$. In data-driven control settings, it is likely to estimate the coefficients $c_1,c_2$ of proportional error bound \cite{lin2025optimalitydeviationusingkoopman}, which be illustrated in Section~\ref{6.numerical}.

Hamilton-Jacobi-Bellman (HJB) and Hamilton-Jacobi-Issac (HJI) equations~\cite{1995Linear,aliyu2011nonlinear} serve as fundamental analytical tools throughout this paper. In the analysis of optimality deviation and optimality robustness, we assume that the considered Hamilton-Jacobi equations admit sufficiently regular classical solutions, so that associated value functions and their gradients are well defined. Under this standing assumption, our focus is on characterizing impacts of uncertainties on optimality rather than analyzing the existance and uniqueness of the solution.

\section{Approximation Error and Optimality Deviation}\label{3.approximation error}
This paper focuses on the following infinite-horizon optimal control problem of \eqref{original system} under the quadratic performance index, i.e.,
\begin{equation}\label{optimal control problem}
	\min_{u(\cdot)} \ J(u(\cdot))=\int_{0}^{\infty}\frac{1}{2}\left[x^\top(t) \overline{Q}x(t) +u^\top(t) Ru(t)\right] \mathrm{d}t
\end{equation}
where $\overline{Q}\succ 0, R\succ 0$ are positive definite. A common scenario in data-driven control applications is that one has no exact knowledge of the approximation error term $r(z,u)$, making it challenging to design the optimal controller for the actual bilinear system \eqref{lifted system AB} that takes the error into consideration, and equivalently for the original nonlinear system \eqref{original system}. Under this circumstance, an ideal approach is to calculate the optimal controller for \eqref{lifted system AB} without considering the error, i.e.,
\begin{subequations}\label{nominal optimal control problem}
	\begin{align}
		V_0^*(z)=\min_{u(\cdot)} \int_{0}^{\infty}\frac{1}{2}\left[z^\top(t) Qz(t) +u^\top(t) Ru(t)\right] \mathrm{d}t \\
		\text{s.t.}\ \dot{z}(t)=Az(t)+B(z(t))u(t), \ z(0)=z \label{nominal trajectory}
	\end{align}
\end{subequations}
where for ease of representation we denote 
\begin{equation}
	Q=C^\top\overline{Q}C,\ B(z)=B_0+\sum_{i=1}^mB_ize_i^\top.
\end{equation}
With the well-known HJB equation \cite{1995Linear}, the \textit{nominal} optimal value function $V_0^*(z)$ is calculated with
\begin{equation}\label{nominal HJB}
	(\nabla V_0^*)^\top Az+\frac{1}{2}z^\top Qz-\frac{1}{2}(\nabla V_0^*)^\top B(z)R^{-1}B^\top(z)\nabla V_0^*=0
\end{equation}
leading to the \textit{nominal} optimal controller
\begin{equation}\label{nominal optimal controller}
	u_0^*(z)=-R^{-1}B^\top(z)\nabla V_0^*.
\end{equation}
There exist a considerable number of works on solving HJB equation and optimal control problem for bilinear systems \cite{1995Linear,TAC2023BilinearSolution}. Further developing suitable and efficient bilinear optimal control methods is a promising direction for Koopman-based data-driven optimal control.

Since the actual system dynamics \eqref{lifted system AB} contains an additional approximation error, applying $u_0^*$ to the actual system will result in performance deviation no matter how advanced the calculation method is. Define the set $\mathcal{R}$ of admissible approximation error satisfying \eqref{error bound}, i.e., 
\begin{equation}
	\mathcal{R}=\{r(z,u)|\|r(z,u)\|\leq c_1\|z\|+c_2\|u\|\}.
\end{equation}
Then we make the following standard assumption.
\begin{assumption}\label{Assum: nominal HJB}
	The nominal optimal controller $u_0^*(z)$ given by \eqref{nominal optimal controller} and the actual optimal controller $u^*(z)$ for the original nonlinear system \eqref{original system} are admissible, i.e., they asymptotically stabilize the lifted bilinear system \eqref{lifted system AB} and yield finite performance indices for admissible uncertainty $r(z,u)\in\mathcal{R}$.
\end{assumption}

First, we consider the optimality deviation in nominal value function $V_0^*(z)$ starting with the following result.
\begin{proposition}\label{Prop: performance deviation}
	Due to the existence of approximation error, applying the nominal optimal controller $u_0^*$ given by \eqref{nominal HJB} and \eqref{nominal optimal controller} to the actual system \eqref{lifted system AB} (equivalently the original nonlinear system \eqref{original system}) results in an extra cost, characterized by
	\begin{equation}\label{optimal performance deviation}
		\begin{aligned}
			\left\|V-V_0^*\right\|\leq \Delta V_{\max}=\frac{1}{2}C_{12}^2\int_{0}^\infty\|\nabla V_0^*\|^2 \mathrm{d}t \\
			+\frac{C_{12}}{2}\left(\int_{0}^\infty\|\nabla V_0^*\|^2 \mathrm{d}t\right)^{\frac{1}{2}}\sqrt{C_{12}^2\int_{0}^\infty\|\nabla V_0^*\|^2 \mathrm{d}t+4V_0^*}.
		\end{aligned}
	\end{equation}
	where $C_{12}=\max\left\{\frac{2c_1L_p}{\sqrt{\lambda_{\min}(\overline{Q})}}, \frac{2c_2}{\sqrt{\lambda_{\min}(R)}}\right\}$, $V(z)$ denotes the value function corresponding to the closed-loop system \eqref{original system} controlled by $u_0^*$, and the integral term is evaluated along the trajectory controlled by $u_0^*$ under the worst-case error $r_0^*$ given by \eqref{nominal worst-case error}, i.e.,
	\begin{equation}\label{nominal worst trajectory u0}
		\dot{z}(t)=A(z(t))+B(z(t))u_0^*(z(t))+r_0^*(z(t),u_0^*),\ z(0)=z.
	\end{equation}
\end{proposition}
Here we use the notation $J(u,z,r)$ since $J$ is actually a function of the initial state $z(0)=z$, the control input $u(t)$ and the approximation error term $r(z(t),u(t))$. Apparently $V_0^*(z)=J(u_0^*,z,0)$ and $V(z)=J(u_0^*,z,r)$. Based on Proposition~\ref{Prop: performance deviation}, we can further analyze the deviation in controllers. The complete proofs of Propositions~\ref{Prop: performance deviation} and \ref{Prop: controller deviation} are omitted due to page limit (see \cite{lin2025optimalitydeviationusingkoopman} for detail).
\begin{proposition}\label{Prop: controller deviation}
	Due to the existence of approximation error, the nominal optimal controller $u_0^*$ given by \eqref{nominal HJB} and \eqref{nominal optimal controller} deviates from the actual optimal controller $u^*$ corresponding to the original nonlinear system \eqref{original system}.  Specifically, this deviation is characterized by
	\begin{equation}\label{controller deviation}
		\begin{aligned}
			&\int_{0}^{\infty}(u_0^*-u^*)^\top R(u_0^*-u^*) \mathrm{d}t \\
			\leq &2\Delta V_{\max}+2C_{12}\left[(V_0^*+\Delta V_{\max})\int_{0}^\infty\|\nabla V_0^*\|^2\mathrm{d}t\right]^\frac{1}{2}
		\end{aligned}
	\end{equation}
	where the integral term is evaluated along the actual optimal trajectory, i.e., \eqref{original system} controlled by $u^*$, $C_{12}$ and $\Delta V_{\max}$ are given by Proposition~\ref{Prop: performance deviation}.
\end{proposition}
\begin{remark}[Conservatism of LTI lifting]\label{Rem: LTI lifting}
	Although many existing works on Koopman-based data-driven control consider LTI lifted models $\dot{z}=Az+B_0u$, the theoretical justification of LTI lifting is limited (see \cite{goswami2021bilinearization,STRASSER2026101035}). Further, the conservatism of LTI lifting can be illustrated by the above analysis. Suppose the approximation error using lifted bilinear models is bounded by $\|r_b(z,u)\|\leq c_{b,1}\|z\|+c_{b,2}\|u\|$. For LTI lifting, an additional bilinear term $r_l(z,u) =\sum_{i=1}^mu_iB_iz$ is absorbed into the error, and one may determine coefficients $c_{l,1},c_{l,2}$ such that $\|r_l(z,u)\|\leq c_{l,1}\|z\|+c_{l,2}\|u\|$. The overall error coefficients $c_1=c_{b,1}+c_{l,1},c_2=c_{b,2}+c_{l,2}$ might be substantially enlarged, especially when $z$ or $u$ is moderately large. In view of Proposition~\ref{Prop: performance deviation} where the optimality deviation depends on $c_1,c_2$, this accumulation of error indicates that linear lifting may induce significantly amplified performance degradation compared with bilinear representations.
\end{remark}
\begin{remark}[Underlying novel analytical strategy]
	The proof of Proposition~\ref{Prop: controller deviation} in \cite{lin2025optimalitydeviationusingkoopman} also confirms that the deviation between the actual value function $V^*=J(z,u^*,r)$ and $V_0^*$ also satisfies $V^*-V_0^*\leq \Delta V_{\max}$. Note that as directly connecting $V_0^*(z)$ and $V^*(z)$, $u_0^*$ and $u^*$ is relatively challenging, the intermediate worst-case value function $V_r^*(z)$ serves as a pivotal bridge for the optimality deviation analysis. See \cite{lin2025optimalitydeviationusingkoopman} for detailed discussions.
\end{remark}

With the above analysis, we have made the abstract notion of \textit{optimality deviation} concrete and numerically computable. The derived upper bounds \eqref{optimal performance deviation}, \eqref{controller deviation} are explicitly parameterized by the error coefficients $(c_1,c_2)$ and nominal optimal value function $V_0^*$. As the closed-loop trajectories can be simulated, the resulting bounds can be efficiently computed offline, as demonstrated in Section~\ref{6.numerical}. This structure directly links the approximation error to the resulting performance loss, providing a practical pathway to estimate these deviations when the exact form of the approximation error is unknown.

\section{Bounded Noise and Optimality Deviation}\label{4.noise}
The Koopman operator has been widely used in data-driven control of nonlinear systems, while the collected data are often corrupted by noise. In this section, we will analyze the impact of noise in data collection on the optimality deviation based on the bounded-energy assumption. It is quite interesting that the following analysis successfully links the impact of bounded noise with the approximation error coefficient $c_1,c_2$, thereby integrating different kinds of uncertainties into a unified framework to analyze and correct the optimality deviation.

\subsection{System Identification with Noisy Data}
First we clarify that the optimal control problem we are interested in remains unchanged, i.e., our objective is still to minimize the quadratic performance index given by \eqref{optimal control problem} for the same unknown system \eqref{original system}. However, noisy data $\left\{x(t_j), u(t_j), \dot{x}(t_j)\right\}_{j=0}^{T-1}$ are collected from noise-corrupted trajectories, described by
\begin{equation*}
	\dot{x}(t)=f\left(x(t)\right)+\sum_{i=1}^{m}g_i\left(x(t)\right)u_i(t)+\overline{d}(t),
\end{equation*}
where $\overline{d}(t)$ represents the effect of noise in data collection rather than an external disturbance acting on system dynamics. With the dictionary function $\Psi(x)$, we obtain
\begin{equation*}
	\frac{\mathrm{d}z}{\mathrm{d}t}=Az+B_0u+\sum_{i=1}^m u_iB_iz+r(z,u)+d(t)
\end{equation*}
where $A,B_0,B_i,i=\{1,\ldots,m\}$ are identified from data and $d(t)=\frac{\partial\Psi(x)}{\partial x}\cdot \overline{d}(t)$ denotes the noise-induced modeling residual. The state measurements and corresponding time derivatives can also be computed via $z(t_j)=\Psi(x(t_j))$ and $\dot{z}(t_j)=\frac{\partial\Psi(x)}{\partial x}\cdot \dot{x}(t_j)$. Construct matrices
\begin{equation}\label{data set}
	\begin{aligned}
		U_0:&=\left[u(t_0) \ u(t_1) \ \ldots \ u(t_{T-1})\right] \\
		Z_0:&=\left[z(t_0) \ z(t_1) \ \ldots \ z(t_{T-1})\right] \\
		V_0^i:&=\left[u_i(t_0)z(t_0) \ u_i(t_1)z(t_1) \ \ldots \ u_i(t_{T-1})z(t_{T-1})\right] \\
		Z_1:&=\left[\dot{z}(t_0) \ \dot{z}(t_1) \ \ldots \ \dot{z}(t_{T-1})\right] \\
		W_0:&=\left[Z_0^\top \ U_0^\top \ (V_0^1)^\top \ \ldots \ (V_0^m)^\top\right]^\top
	\end{aligned}
\end{equation}
where $u_i(t_j)$ denotes the $i$-th component of $u(t_j)$. Further, we arrange the unknown noise sequence as
\begin{equation}
	D_0:=\left[d(t_0) \ d(t_1) \ \ldots \ d(t_{T-1})\right].
\end{equation}
\begin{assumption}\label{Assum: Disturbance}
	Without loss of generality, we assume: \\
	(a) The matrix $W_0\in\mathbb{R}^{N\times T}$ is of full row rank. \\
	(b) $D_0\in\mathcal{D}$ such that for some matrix $\Delta$,
	\begin{equation}\label{disturbance energy}
		\mathcal{D}:=\left\{D\in\mathbb{R}^{N\times T}:DD^\top\preceq\Delta\Delta^\top\right\}.
	\end{equation}
\end{assumption}
\begin{remark}[Reasonability of assumption]
	Here we make a brief explanation. Assumption~\ref{Assum: Disturbance}.(a) is similar to the notion \textit{persistence of excitation} in data-driven control of linear systems, which promises the quality of data and the uniqueness of least square solution in system identification (see \eqref{least square problem}). Assumption~\ref{Assum: Disturbance}.(b) exhibits a general and widely used energy bound of noise or disturbance \cite{BISOFFI2022Petersen}.
\end{remark}

In Koopman-based data-driven control, EDMD is widely used for system identification, with basic form \cite{2017On}
\begin{equation}\label{least square problem}
	\min_{A,B_0,B_1,\ldots,B_m}\left\|Z_1-\left[A\ B_0\ B_1\ \ldots\ B_m\right]W_0\right\|
\end{equation}
which is a least square problem solved by
\begin{equation}\label{identified Z}
	\left[{A}\ {B}_0\ {B}_1\ \ldots\ {B}_m\right]:={\mathbf{Z}}^\top =Z_1W_0^\dagger.
\end{equation}
The approximation error of EDMD, including projection and estimation error, is interpreted in $r(z,u)$. When considering the impact of noise, least-square solution \eqref{identified Z} deviates from matrices corresponding to the actual lifted dynamics, which should be theoretically solved by
\begin{equation}\label{actual Z}
	\left[\tilde{A}\ \tilde{B}_0\ \tilde{B}_1\ \ldots\ \tilde{B}_m\right]:=\tilde{\mathbf{Z}}^\top =(Z_1-D_0)W_0^\dagger.
\end{equation}
Since we have no exact knowledge of $D_0$ but the energy bound $\Delta\Delta^\top$ is known a priori, we should investigate all the system matrices \eqref{actual Z} satisfying $D_0\in \mathcal{D}$, rather than the single solution \eqref{identified Z} ignoring the impact of noise. This leads to the set $\mathcal{Z}_0$ of matrices \textit{consistent with noisy data}
\begin{equation}
	\mathcal{Z}_0:=\left\{\tilde{\mathbf{Z}}^\top=(Z_1-D_0)W_0^\dagger: D_0\in \mathcal{D}\right\},
\end{equation}
i.e., the set of all pairs of matrices $[A\ B_0\ B_1\ \ldots\ B_m]$ that can generate the noisy data $W_0,Z_1$.

\subsection{From Noise Bound to Optimality Deviation}
The lifted bilinear system \eqref{lifted system AB} is identified via EDMD algorithm \eqref{least square problem}-\eqref{identified Z}, but the solution might be inaccurate due to the existence of noise. However, all of the possible matrices corresponding to the actual lifted dynamics are recorded in $\mathcal{Z}_0$. We first introduce the following result.
\begin{proposition}\label{Prop: Set Reformulation}
	Define the sets
	\begin{subequations}
		\begin{align}
			&\mathcal{Z}:=\left\{\tilde{\mathbf{Z}}^\top: (\tilde{\mathbf{Z}}- \boldsymbol{\zeta})^\top \mathbf{A}(\tilde{\mathbf{Z}}-\boldsymbol{\zeta})\preceq \mathbf{Q} \right\} \\
			&\mathcal{E}:=\left\{(\boldsymbol{\zeta}+\mathbf{A}^{-\frac{1}{2}}\boldsymbol{\gamma}\mathbf{Q}^\frac{1}{2})^\top:\|\boldsymbol{\gamma}\|\leq 1\right\}
		\end{align}
	\end{subequations}
	where $\mathbf{A}=W_0W_0^\top,\boldsymbol{\zeta}=(Z_1W_0^\dagger)^\top=(W_0W_0^\top)^{-1} W_0Z_1^\top$ and $\mathbf{Q}=\Delta\Delta^\top$. Then $\mathcal{Z}_0\subseteq\mathcal{Z} =\mathcal{E}$.
\end{proposition}
See Appendix~\ref{Proof: Set Reformulation} for the proof. Proposition~\ref{Prop: Set Reformulation} is a further reformulation of the matrix set $\mathcal{C}_0$ consistent with noisy data. Subsequently, the impact of noise is interpreted as a bounded perturbation of $\boldsymbol{\zeta}$ which corresponds to the least-square solution \eqref{identified Z} of bilinear system identification under noise-free case. Therefore, it is convenient to incorporate the impact of noise into error term $r(z,u)$.
\begin{theorem}\label{Thm: Noise to Error}
	Due to the existence of approximation error and noise in data collection, there is deviation between the identified bilinear system using \eqref{identified Z} and the actual lifted bilinear system given by \eqref{lifted system AB}. The overall error $r(z,u)=r_a(z,u)+r_d(z,u)$ captures the impacts of approximation error and noise, while 
	\begin{equation}\label{noise error bound}
		\|r_d(z,u)\|\leq c_d\left(\|z\|+\|u\|+\|z\|\|u\|\right)
	\end{equation}
	holds where the noise-related coefficient is defined as
	\begin{equation}
		c_d=\left\|\mathbf{Q}^{\frac{1}{2}}\right\|\left\|\mathbf{A}^{-\frac{1}{2}}\right\| =\left\|(\Delta\Delta^\top)^{\frac{1}{2}}\right\|\left\|(W_0W_0^\top)^{-\frac{1}{2}}\right\|.
	\end{equation}
\end{theorem}
\begin{pf}
	We have demonstrated that the impact of noise in data collection lies in the perturbation of identified system matrices $\mathbf{Z}^\top=[A\ B_0\ B_1\ \ldots \ B_m]$. Here we use notation $[\Delta A\ \Delta B_0\ \Delta B_1\ \ldots \ \Delta B_m]$ for the perturbation. Hence, noise-induced deviation between identified and actual lifted bilinear systems is characterized by
	\begin{equation*}
		r_d(z,u)=\Delta A\cdot z+\Delta B_0\cdot u+[\Delta B_1z\ \ldots\ \Delta B_mz]u,
	\end{equation*}
	which can be bounded by
	\begin{equation}\label{noise norm derivation}
		\begin{aligned}
			\|r_d(z,u)\|&\leq \|\Delta A\|\|z\|+\|\Delta B_0\|\|u\| \\
			&+\|[\Delta B_1z\ \ldots\ \Delta B_mz]\|\|u\|.
		\end{aligned}
	\end{equation}
	Since $[A+\Delta A\ B_0+\Delta B_0\ B_1+\Delta B_1\ \ldots \ B_m+\Delta B_m]\in\mathcal{C}_0 \subseteq \mathcal{C}=\mathcal{E}$, there exists a $\boldsymbol{\gamma}$, such that $\|\boldsymbol{\gamma}\|\leq 1$ and
	\begin{equation*}
		[\Delta A\ \Delta B_0\ \Delta B_1\ \ldots \ \Delta B_m]=\mathbf{Q}^{\frac{1}{2}} \boldsymbol{\gamma}^\top \mathbf{A}^{-\frac{1}{2}}.
	\end{equation*}
	With the definition of Frobenius norm, all three norms, $\|\Delta A\|$, $\|\Delta B_0\|$, $\|[\Delta B_1\ \ldots\ \Delta B_m]\|$, are no greater than
	\begin{equation*}
		\left\|[\Delta A\ \Delta B_0\ \Delta B_1\ \ldots \ \Delta B_m]\right\|\leq \left\|\mathbf{Q}^{\frac{1}{2}}\right\|\cdot\left\|\mathbf{A}^{-\frac{1}{2}}\right\|=c_d.
	\end{equation*}
	Additionally, since
	\begin{equation*}
		\begin{aligned}
			\left\|[\Delta B_1z\ \ldots\ \Delta B_mz]\right\|^2=\sum_{i=1}^m\|\Delta B_iz\|^2&\leq \|z\|^2 \sum_{i=1}^m\|\Delta B_i\|^2\\
			=\|z\|^2\left\|[\Delta B_1\ \ldots\ \Delta B_m]\right\|^2 &\leq c_d^2\|z\|^2
		\end{aligned}
	\end{equation*}
	together with \eqref{noise norm derivation}, the upper bound for \eqref{noise error bound} holds. \qed
\end{pf}

In Koopman-based data-driven control practice, we usually consider a compact, forward-invariant state space $x\in\mathbb{X}$ without compromising the existence of optimal controller (or the solution of HJB/HJI equation). Then, \eqref{noise error bound} can be over-approximated by proportional error bound similar to \eqref{error bound}, e.g.,
\begin{equation}
	\|r_d(z,u)\|\leq c_d\left(1+\max_{x\in\mathbb{X}}\|\Psi(x)\|\right)\|z\|+c_d\|u\|.
\end{equation}
This observation explains why the impact of noise in data collection can be incorporated into the error bound \eqref{error bound}. Accordingly, in the following text we will continue our discussions based on the proportional error bound \eqref{error bound}, which is assumed to capture the impacts of all uncertainties. This formulation is justified by the existence of $c_{a,1},c_{a,2}>0$ such that
\begin{equation}
	\|r_a(z,u)\|\leq c_{a,1}\|z\|+c_{a,2}\|u\|.
\end{equation}
Let $c_1=c_{a,1}+c_d\left(1+\max_{x\in\mathbb{X}}\|\Psi(x)\|\right)$ and $c_2=c_{a,2}+c_d$, then the impacts of approximation error and noise in data collection are jointly captured by $r(z,u)$ bounded by \eqref{error bound}. Nevertheless, this treatment might introduce certain conservatism as it relies on a bounded state space and neglects part of the structural information in $[\Delta A\ \Delta B_0\ \Delta B_1\ \ldots \ \Delta B_m]$. Exploring less conservative characterizations for noise or disturbance is an interesting direction for future work.

\section{Correcting the Optimality Deviation}\label{5.robust design}
Building upon the characterization of optimality deviations due to the existence of uncertainties, a natural yet critical question arises: how can such deviations be systematically mitigated at the controller design stage? The analysis in Section~\ref{3.approximation error} has revealed the \textit{worst-case} impact of approximation error, and we have illustrated that noise in data collection can be unified within the same analytical framework. These results motivate a shift from conventional stability robustness to an \textit{optimality robustness} perspective.

From this viewpoint, robustness is interpreted in terms of the magnitude of optimality deviations and the controller’s capability to mitigate them in the presence of uncertainties. Consequently, mitigating optimality deviations amounts to designing controllers that explicitly counteract the worst-case impacts of uncertainties, which naturally leads to a robust optimal control formulation in a min-max optimization form
\begin{subequations}\label{robust optimal control problem}
	\begin{align}
		V_{ro}^*(z)=\min_{u(\cdot)}\max_{r\in\mathcal{R}} \int_{0}^{\infty}\frac{1}{2}\left[z^\top(t) Qz(t) +u^\top Ru\right] \mathrm{d}t \\
		\text{s.t.}\ \dot{z}(t)=Az(t)+B(z(t))u+r(z(t),u), \ z(0)=z. \label{actual dynamics}
	\end{align}
\end{subequations}
Using HJI equation \cite{aliyu2011nonlinear}, the robust optimal value function $V_{ro}^*$ should be solved from
\begin{equation*}
	\begin{aligned}
		\min_{u}\max_{r\in\mathcal{R}}\{(\nabla V_{ro}^*)^\top[Az+B(z)u +r(z,u)] \\ +\frac{1}{2}z^\top Qz+\frac{1}{2}u^\top Ru\}=0.
	\end{aligned}
\end{equation*}
Therefore, the worst-case approximation error $r^*$ (similar to $r_0^*$) and robust optimal controller $u_{ro}^*$ satisfy
\begin{equation}\label{robust optimal controller}
	\begin{aligned}
		&r^*(z,u_{ro}^*)=\frac{c_1\|z\|+c_2\|u_{ro}^*\|}{\|\nabla V_{ro}^*\|}\nabla V_{ro}^*,\\
		&u_{ro}^*(z)=-R^{-1}B^\top(z)\nabla V_{ro}^*-\frac{c_2\|\nabla V_{ro}^*\|}{\|u_{ro}^*\|}R^{-1}u_{ro}^*.
	\end{aligned}
\end{equation}
where the robust optimal value function is solved by
\begin{equation}\label{robust HJI}
	\begin{aligned}
		0=\left(\nabla V_{ro}^*\right)^\top Az+\frac{1}{2}z^\top Qz-\frac{1}{2}(\nabla V_{ro}^*)^\top B(z)R^{-1}B^\top(z)\nabla V_{ro}^* \\ 
		+\left(c_1\|z\|+c_2\|u_{ro}^*\|\right)\|\nabla V_{ro}^*\|+\frac{c_2^2 \|\nabla V_{ro}^*\|^2}{2\|u_{ro}^*\|^2}\left(u_{ro}^*\right)^\top R^{-1}u_{ro}^*.
	\end{aligned}
\end{equation}
It should be noted that the robust optimal controller \eqref{robust optimal controller} is given in an implicit form and represents the first-order necessary optimality condition. Cases where the optimal control or value function gradient vanishes typically occur at the isolated equilibrium points, hence we impose $r^*(0,0)=0$ to ensure that $r^*$ is continuous and well-defined. Furthermore, we make the following assumption similar with Assumption~\ref{Assum: nominal HJB}, which is a routinely adopted admissibility condition in nonlinear robust optimal and data-driven control formulations, see, e.g., \cite{aliyu2011nonlinear,JZP2022TNNLSContraction}.
\begin{assumption}\label{Assum: robust HJI}
	The robust optimal controller $u_{ro}^*(z)$ is admissible in the sense of robust control, i.e., it asymptotically stabilizes the lifted bilinear system \eqref{lifted system AB} and yields a finite performance index for admissible uncertainty $r(z,u)\in\mathcal{R}$.
\end{assumption}
The remainder of this section addresses two significant aspects. Firstly, we quantify to what extent robust controller design $u_{ro}^*$ via the above methodology \eqref{robust optimal controller}-\eqref{robust HJI} can correct the optimality deviations, where the analyses are dual to Propositions~\ref{Prop: performance deviation} and \ref{Prop: controller deviation}. Secondly, since it is nontrivial to obtain an analytical solution of \eqref{robust optimal controller}-\eqref{robust HJI}, we introduce a policy‑iteration algorithm, thereby bridging theoretical guarantees with computational tractability.

\subsection{Theoretical Analysis}
Intuitively speaking, adopting controller $u_{ro}^*$ sacrifices the nominal optimality to a certain degree under the ideal case, i.e., $r(z,u)=0$. However, it guarantees the performance improvement to the maximum degree under the worst-case uncertainties. Consequently, a trade-off between nominal performance and optimality robustness lies in the robust optimal control methodology.
\begin{theorem}\label{Thm: Robust Performance}
	Under the ideal error-free case, adopting the robust optimal controller $u_{ro}^*$ results in an extra cost
	\begin{equation}\label{extra cost uro}
		J(u_{ro}^*,z,0)-J(u_0^*,z,0)=\frac{1}{2}\int_{0}^{\infty}(u_{ro}^*-u_0^*)^\top R(u_{ro}^*-u_0^*) \mathrm{d}t
	\end{equation}
	where the integral term is evaluated along the nominal trajectory \eqref{nominal trajectory} controlled by $u_{ro}^*$. Conversely, under the worst-case approximation error given by \eqref{robust optimal controller}, adopting the nominal optimal controller $u_0^*$ results in an extra cost
	\begin{equation}\label{performance improvement uro}
		\begin{aligned}
			J(u_0^*&,z,r^*)-J(u_{ro}^*,z,r^*) \\
			\leq\int_{0}^{\infty}\big[&(1+\frac{1}{2\lambda_{\min}(R)})(u_{ro}^*-u_0^*)^\top R(u_{ro}^*-u_0^*)\\ +&(1+\frac{1}{\lambda_{\min}(R)})\frac{c_2^2}{2} \|\nabla V_{ro}^*\|^2\big]\mathrm{d}t
		\end{aligned}
	\end{equation}
	where the integral term is evaluated along the trajectory controlled by $u_0^*$ under the worst-case error $r^*$, i.e.,
	\begin{equation}\label{worst trajectory u0}
		\dot{z}(t)=A(z(t))+B(z(t))u_0^*(z(t))+r^*(z(t),u_0^*),\ z(0)=z.
	\end{equation}
\end{theorem}
\begin{pf}
	Along the system trajectory \eqref{nominal trajectory} controlled by $u_{ro}^*$, the time derivative of $V_0^*=J(u_0^*,z,0)$ is given by
	\begin{equation*}
		\frac{\mathrm{d}V_0^*}{\mathrm{d}t}=(\nabla V_0^*)^\top\left(Az(t)+B(z(t))u_{ro}^*(z(t)) \right).
	\end{equation*}
	Since $V_0^*$ is solved with \eqref{nominal HJB}, we further obtain
	\begin{equation*}
		\frac{\mathrm{d}V_0^*}{\mathrm{d}t}=-\frac{1}{2}z^\top(t)Qz(t)+\frac{1}{2}(u_0^*)^\top R u_0^*-(u_0^*)^\top Ru_{ro}^*.
	\end{equation*}
	With Assumption~\ref{Assum: robust HJI}, the extra cost satisfies
	\begin{equation*}
		\begin{aligned}
			&J(u_{ro}^*,z,0)-J(u_0^*,z,0) \\
			=&\int_{0}^{\infty} \frac{1}{2}z^\top(t)Qz(t) +\frac{1}{2}(u_{ro}^*)^\top Ru_{ro}^*+\frac{\mathrm{d}V_0^*}{\mathrm{d}t} \mathrm{d}t \\
			=&\frac{1}{2}\int_{0}^{\infty} (u_{ro}^*-u_0^*)^\top R(u_{ro}^*-u_0^*)\mathrm{d}t.
		\end{aligned}
	\end{equation*}
	Along \eqref{worst trajectory u0}, the time derivative of $V_{ro}^*=J(u_{ro}^*,z,r^*)$ is
	\begin{equation*}
		\frac{\mathrm{d}V_{ro}^*}{\mathrm{d}t}=(\nabla V_{ro}^*)^\top \left(Az(t)+B(z(t))u_0^*+r^*(z(t),u_0^*)\right).
	\end{equation*}
	Since $V_{ro}^*$ is solved with \eqref{robust HJI}, we further calculate and simplify the time derivative, obtaining
	\begin{equation*}
		\begin{aligned}
			\frac{\mathrm{d}V_{ro}^*}{\mathrm{d}t}=-\|\nabla V_{ro}^*\|\left(c_1\|z\|+c_2\|u_{ro}^*\|\right)+(\nabla V_{ro}^*)^\top r^*(z(t),u_0^*) \\ +(\nabla V_{ro}^*)^\top B(z(t))u_0^*-\frac{c_2^2 \|\nabla V_{ro}^*\|^2}{2\|u_{ro}^*\|^2}\left(u_{ro}^*\right)^\top R^{-1}u_{ro}^* \\
			+\frac{1}{2}(\nabla V_{ro}^*)^\top B(z(t))R^{-1}B^\top(z(t))\nabla V_{ro}^* -\frac{1}{2}z^\top(t)Qz(t)
		\end{aligned}
	\end{equation*}
	Since Assumption~\ref{Assum: robust HJI} admits that $V_{ro}^*(z(\infty))=0$, the extra cost satisfies
	\begin{equation*}
		\begin{aligned}
			&J(u_0^*,z,r^*)-J(u_{ro}^*,z,r^*) \\
			=&\int_{0}^{\infty} \frac{1}{2}z^\top(t)Qz(t) +\frac{1}{2}(u_0^*)^\top Ru_0^*+\frac{\mathrm{d}V_{ro}^*}{\mathrm{d}t} \mathrm{d}t \\
			=&\int_{0}^{\infty} \big[\frac{1}{2}\nabla(V_{ro}^*-V_0^*)^\top B(z(t))R^{-1}B^\top(z(t)) \nabla (V_{ro}^*-V_0^*) \\
			+&c_2\|\nabla V_{ro}^*\|(\|u_0^*\|-\|u_{ro}^*\|)-\frac{c_2^2 \|\nabla V_{ro}^*\|^2}{2\|u_{ro}^*\|^2}\left(u_{ro}^*\right)^\top R^{-1}u_{ro}^*\big]\mathrm{d}t.
		\end{aligned}
	\end{equation*}
	With the form of $u_0^*$ and $u_{ro}^*$ in \eqref{nominal optimal controller} and \eqref{robust optimal controller}, the extra cost is simplified and bounded with Cauchy-Schwartz inequality by
	\begin{equation}\label{robust integrand}
		\begin{aligned}
			&J(u_0^*,z,r^*)-J(u_{ro}^*,z,r^*) \\
			=&\int_{0}^{\infty}\big[\frac{c_2\|\nabla V_{ro}^*\|}{\|u_{ro}^*\|}(u_{ro}^*)^\top R^{-1}(u_{ro}^*-u_0^*) \\
			+&\frac{1}{2}(u_0^*-u_{ro}^*)^\top R(u_0^*-u_{ro}^*)+c_2\|\nabla V_{ro}^*\|(\|u_0^*\|-\|u_{ro}^*\|) \big] \mathrm{d}t \\
			\leq& \int_{0}^{\infty}\big[\frac{1}{2}c_2^2 \|\nabla V_{ro}^*\|^2\cdot \frac{\left(u_{ro}^*\right)^\top R^{-1}u_{ro}^*}{\|u_{ro}^*\|^2} \\
			+&(u_0^*-u_{ro}^*)^\top R(u_0^*-u_{ro}^*) +c_2\|\nabla V_{ro}^*\|\|u_0^*-u_{ro}^*\|\big]\mathrm{d}t.
		\end{aligned}
	\end{equation}
	Using Cauchy-Schwartz inequality again, the last term of the integrand in \eqref{robust integrand} is bounded with
	\begin{equation}\label{robust integrand 3}
		\begin{aligned}
			c_2\|\nabla V_{ro}^*\|\|u_0^*-u_{ro}^*\|\leq \frac{1}{2}c_2^2\|\nabla V_{ro}^*\|^2 +\frac{1}{2}\|u_0^*-u_{ro}^*\|^2 \\
			\leq \frac{1}{2}c_2^2\|\nabla V_{ro}^*\|^2+\frac{1}{2\lambda_{\min}(R)} (u_0^*-u_{ro}^*)^\top R(u_0^*-u_{ro}^*).
		\end{aligned}
	\end{equation}
	Meanwhile, the first term of the integrand in \eqref{robust integrand} can be bounded using the smallest eigenvalue of $R$, since
	\begin{equation}\label{robust integrand 1}
		\frac{\left(u_{ro}^*\right)^\top R^{-1}u_{ro}^*}{\|u_{ro}^*\|^2} \leq \lambda_{\max}(R^{-1})=\frac{1}{\lambda_{\min}(R)}.
	\end{equation}
	Combining \eqref{robust integrand}, \eqref{robust integrand 3} and \eqref{robust integrand 1} completes the proof.
\end{pf}
\begin{remark}[Balancing performance and robustness]\label{Rem: Tradeoff}
	As stated above, Theorem~\ref{Thm: Robust Performance} reveals a clear \textit{performance-robustness trade-off} inherent in the proposed methodology. Specifically, the cost of robustness under the error-free case, given by \eqref{extra cost uro}, is solely attributed to the integrated deviation between the robust optimal controller $u_{ro}^*$ and nominal one $u_0^*$. By contrast, when the approximation error is present, the performance degradation of $u_0^*$ given by \eqref{performance improvement uro}, as well as the performance improvement achieved by $u_{ro}^*$ under the worst-case error, exhibits a fundamentally different structure. Particularly, it not only amplifies the contribution of controller deviation, but also includes an additional positive term proportional to $c_2^2\|\nabla V_{ro}^*\|^2$. Although the integrals follow different trajectories, this structural gap reveals that the robust controller design deliberately accepts a quantifiable and often small nominal performance loss to secure a guaranteed and potentially large improvement of optimality robustness in the presence of approximation error.
\end{remark}

Meanwhile, the comparison from another perspective of controller is also necessary, i.e., between $u_{ro}^*$ and the actual optimal controller $u^*$.
\begin{theorem}\label{Thm: Robust Controller}
	Suppose that the actual optimal controller $u^*$ stabilizes the original nonlinear system \eqref{original system}. Then, the deviation between $u^*$ and the robust optimal controller $u_{ro}^*$ is characterized by
	\begin{equation}\label{robust controller deviation}
		\begin{aligned}
			&\int_{0}^{\infty}(u_{ro}^*-u^*)^\top R(u_{ro}^*-u^*)\mathrm{d}t \\
			\leq& \max\left\{\frac{4c_1L_p}{\sqrt{\lambda_{\min}(\overline{Q})}}, \frac{4c_2}{\sqrt{\lambda_{\min}(R)}}\right\}\left(V^*\int_{0}^\infty\|\nabla V_{ro}^*\|^2\mathrm{d}t\right)^{\frac{1}{2}}
		\end{aligned}
	\end{equation}
	where the integral term is evaluated along the actual optimal trajectory, i.e., \eqref{original system} controlled by $u^*$.
\end{theorem}
\begin{pf}
	Along the trajectory \eqref{original system} (equivalently \eqref{lifted system AB}) controlled by $u^*$, the time derivative of $V_{ro}^*$ is given by
	\begin{equation*}
		\frac{\mathrm{d}V_{ro}^*}{\mathrm{d}t}=(\nabla V_{ro}^*)^\top \left(Az(t)+B(z(t))u^*+r(z(t),u^*) \right).
	\end{equation*}
	The time derivative of $V^*$ satisfies a similar form. Since the value function $V_{ro}^*$ is solved with \eqref{robust HJI}, we obtain
	\begin{equation*}
		\begin{aligned}
			\frac{\mathrm{d}}{\mathrm{d}t}&(V_{ro}^*-V^*)
			=\frac{1}{2}(\nabla V_{ro}^*)^\top B(z(t))R^{-1}B^\top(z(t))\nabla V_{ro}^* \\
			+&(\nabla V_{ro}^*)^\top B(z(t))u^*+\frac{1}{2}(u^*)^\top Ru^*+(\nabla V_{ro}^*)^\top r(z,u^*) \\
			-&(c_1\|z\|+c_2\|u_{ro}^*\|)\|V_{ro}^*\|-\frac{c_2^2 \|\nabla V_{ro}^*\|^2}{2\|u_{ro}^*\|^2}\left(u_{ro}^*\right)^\top R^{-1}u_{ro}^*.
		\end{aligned}
	\end{equation*}
	With \eqref{robust optimal controller}, we have the following relation
	\begin{equation*}
		B^\top (z)\nabla V_{ro}^*=-R\left(u_{ro}^* +\frac{c_2\|\nabla V_{ro}^*\|}{\|u_{ro}^*\|}R^{-1}u_{ro}^*\right),
	\end{equation*}
	then the time derivative of $V_{ro}^*-V^*$ is written as
	\begin{equation*}
		\begin{aligned}
			\frac{\mathrm{d}}{\mathrm{d}t}&(V_{ro}^*-V^*)
			=(\nabla V_{ro}^*)^\top [r(z,u^*)-\frac{c_1\|z\|+c_2\|u_{ro}^*\|}{\|\nabla V_{ro}^*\|}\nabla V_{ro}^*] \\
			&+\frac{1}{2}(u^*-u_{ro}^*)^\top R(u^*-u_{ro}^*) -c_2\|V_{ro}^*\|\frac{(u_{ro}^*)^\top(u^*-u_{ro}^*)}{\|u_{ro}^*\|} \\
			&=(\nabla V_{ro}^*)^\top \left[r(z,u^*)-\frac{c_1\|z\|+c_2\frac{(u_{ro}^*)^\top u^*}{\|u_{ro}^*\|}}{\|\nabla V_{ro}^*\|}\nabla V_{ro}^*\right] \\
			&+\frac{1}{2}(u^*-u_{ro}^*)^\top R(u^*-u_{ro}^*).
		\end{aligned}
	\end{equation*}
	Integrating the above equation along the actual system trajectory \eqref{lifted system AB} controlled by $u^*$, we obtain
	\begin{equation}\label{robust controller deviation derivation}
		\begin{aligned}
			&V_{ro}^*-V^*=-\frac{1}{2}\int_{0}^\infty(u^*-u_{ro}^*)^\top R(u^*-u_{ro}^*)\mathrm{d}t \\ -&\int_{0}^\infty(\nabla V_{ro}^*)^\top[r(z,u^*) -\frac{c_1\|z\|+c_2\frac{(u_{ro}^*)^\top u^*}{\|u_{ro}^*\|}}{\|\nabla V_{ro}^*\|}\nabla V_{ro}^*] \mathrm{d}t
		\end{aligned}
	\end{equation}
	since we assume that the optimal controller $u^*$ stabilizes the system then $V^*(z(\infty))=V_{ro}^*(z(\infty))=0$.
	
	Recall that $u_{ro}^*$ achieves optimality under the worst-case approximation error $r^*$, in other words,
	\begin{equation}
		V_{ro}^*(z)=J(u_{ro}^*,z,r^*)\geq J(u_{ro}^*,z,r)\geq J(u^*,z,r)=V^*(z).
	\end{equation}
	Hence with \eqref{robust controller deviation derivation}, the controller deviation is bounded by
	\begin{equation*}
		\begin{aligned}
			&\int_{0}^\infty(u^*-u_{ro}^*)^\top R(u^*-u_{ro}^*)\mathrm{d}t \\
			\leq &\int_{0}^\infty(\nabla V_{ro}^*)^\top\left[ \frac{c_1\|z\|+c_2\frac{(u_{ro}^*)^\top u^*}{\|u_{ro}^*\|}}{\|\nabla V_{ro}^*\|}\nabla V_{ro}^*-r(z,u^*)\right] \mathrm{d}t \\
			\leq &\int_{0}^\infty \|\nabla V_{ro}^*\|\cdot 2\left(c_1\|z(t)\|+c_2\|u^*(z(t))\| \right) \mathrm{d}t \\
			\leq &2\left(\int_{0}^\infty \|\nabla V_{ro}^*\|^2 \mathrm{d}t\right)^\frac{1}{2} \left(\int_{0}^\infty (c_1\|z(t)\|+c_2\|u^*\|)^2 \mathrm{d}t\right)^\frac{1}{2}.
		\end{aligned}
	\end{equation*}
	The last integral term can be 
	\begin{equation*}
		\begin{aligned}
			\int_{0}^\infty (c_1\|z\|+c_2\|u^*\|)^2 \mathrm{d}t \leq\int_{0}^\infty2\left(c_1^2\|z\|^2+c_2^2\|u^*\|^2\right)\mathrm{d}t \\
			\leq\int_{0}^\infty2\left(c_1^2L_p^2\|x\|^2+c_2^2\|u_0^*\|^2\right)\mathrm{d}t \leq C_{12}^2V^*
		\end{aligned}
	\end{equation*}
	where $C_{12}$ is given in Proposition~\ref{Prop: performance deviation}, and the upper bound in \eqref{robust controller deviation} holds.
\end{pf}
\begin{remark}[Comparing the controller deviation]
	According to the proof of \cite[Theorem~4.4]{lin2025optimalitydeviationusingkoopman}, the optimality deviation for $u_0^*$ in \eqref{nominal optimal controller} is upper-bounded with
	\begin{equation}\label{controller deviation comparsion}
		\begin{aligned}
			&\int_{0}^{\infty}(u_0^*-u^*)^\top R(u_0^*-u^*)\mathrm{d}t\leq 2\Delta V_{\max}\\
			+&\max\left\{\frac{4c_1L_p}{\sqrt{\lambda_{\min}(\overline{Q}})}, \frac{4c_2}{\sqrt{\lambda_{\min}(R)}}\right\}\left(V^*\int_{0}^\infty\|\nabla V_0^*\|^2\mathrm{d}t\right)^{\frac{1}{2}}.
		\end{aligned}
	\end{equation}
	In contrast, only the last term of \eqref{controller deviation comparsion} appears in \eqref{robust controller deviation}, the upper bound for optimality deviation of $u_{ro}^*$ (definitely $\nabla V_0^*$ is replaced by $\nabla V_{ro}^*$). Although the two bounds cannot be strictly ordered without further assumptions, this structural comparison intuitively highlights that the robust optimal controller is designed to actively compensate for the approximation error within its feedback law, thereby reducing the potentially worst-case optimality deviation.
\end{remark}

\subsection{Practical Implementation}
We have demonstrated that the robust optimal control methodology by solving the min-max problem \eqref{robust optimal control problem} efficiently achieves a performance improvement under the worst-case approximation error. However, solving the coupled equations \eqref{robust HJI} and \eqref{robust optimal controller} directly is intractable due to the nonlinear interdependency between $u_{ro}^*$ and $V_{ro}^*$. To bridge this critical gap between theory and practice, we propose a policy iteration algorithm that alternates between evaluating the cost of a given policy and improving it, inspired by basic principles of reinforcement learning and adaptive dynamic programming \cite{CSM2012RL}.

Due to the first-order and nonlinear nature of the HJB equation, classical solutions may fail to exist in general, which poses substantial challenges for the convergence analysis as well as practical computation for iterative algorithms. To overcome this difficulty, we adopt a vanishing viscosity regularization by introducing a small diffusion term. Specifically, in the policy evaluation step during the $k$-th iteration, we solve the following equation
\begin{equation}\label{policy evaluation}
	\begin{aligned}
		\left(\nabla V_{\varepsilon}^{(k+1)}\right)^\top Az-\frac{1}{2}\left(\nabla V_{\varepsilon}^{(k+1)}\right)^\top B(z)R^{-1}B^\top(z)\nabla V_{\varepsilon}^{(k+1)} \\ 
		+\frac{1}{2}z^\top Qz-\varepsilon \nabla^2 V_{\varepsilon}^{(k+1)} =-\left(c_1\|z\|+c_2\|u_{\varepsilon}^{(k)}\|\right)\|\nabla V_{\varepsilon}^{(k)}\| \\ -\frac{c_2^2 \|\nabla V_{\varepsilon}^{(k)}\|^2}{2\|u_{\varepsilon}^{(k)}\|^2} \left(u_{\varepsilon}^{(k)}\right)^\top R^{-1}u_{\varepsilon}^{(k)}
	\end{aligned}
\end{equation}
where diffusion parameter $\varepsilon>0$ is sufficiently small and $\nabla^2$ denotes the Laplace operator $\nabla^2 V_{ro}=\nabla\cdot(\nabla V_{ro})$. This technique has been explored for analyzing HJB equations to ensure sufficient smoothness while preserving the essential structure of the original problem \cite{bardi1997optimal}. In the policy update step, the controller is updated with the newly obtained value function $V_{ro}^{(k+1)}$ by
\begin{equation}\label{policy improvement}
	u_{\varepsilon}^{(k+1)}(z)=-R^{-1}B^\top(z)\nabla V_{\varepsilon}^{(k+1)}-\frac{c_2\|\nabla V_{\varepsilon}^{(k+1)}\|}{\|u_{\varepsilon}^{(k)}\|}R^{-1}u_{\varepsilon}^{(k)}.
\end{equation}
It is worth noting that the vanishing viscosity term $\varepsilon \nabla^2 V$ is introduced mainly for analytical regularization. As $\varepsilon \to 0$, the corresponding regularized solutions converge to the viscosity solution of HJI equation \eqref{robust HJI}, thereby preserving the structure of underlying robust optimal control methodology \eqref{robust optimal controller} and \eqref{robust HJI}. With $\varepsilon>0$, the policy evaluation equation \eqref{policy evaluation} at each iteration step becomes a quasilinear elliptic PDE, which can be solved using well-established numerical methods \cite{doi:10.1137/110825960}.

The algorithm maintains full compatibility with classical optimal control settings. In the ideal error-free case $r(z,u)=0$ where $c_1=c_2=0$, setting $\varepsilon=0$ recovers the classical HJB equation associated with the nominal optimal control problem, which can be efficiently solved by existing numerical methods \cite{TNNLS2015RL,Chen2018IterativeHJB}. The overall structure of proposed scheme is summarized in Algorithm~1.
\begin{algorithm}[H]\label{Algorithm 1}
	\caption{Policy Iteration Implementation for Koopman-Based Robust Optimal Control Methodology}\label{alg:alg1}
	\begin{algorithmic}
		\State \textbf{Initialization}
		\State \hspace{0.5cm}{Select sufficiently small parameters} $ \varepsilon > 0, \nu > 0 $.
		\State \hspace{0.5cm}{Set $k=0$, select an initial admissible policy $ u_{\varepsilon}^{(0)} $ and corresponding value function $ V_{\varepsilon}^{(0)} $, iterate on $k$.} 
		\State \textbf{Policy Evaluation}
		\State \hspace{0.5cm}{Obtain the value function $V_{\varepsilon}^{(k+1)}$ (or gradient $\nabla V_{\varepsilon}^{(k+1)}$) by solving \eqref{policy evaluation}.}
		\State \textbf{Policy Improvement}
		\State \hspace{0.5cm}{Obtain the updated policy $u_{\varepsilon}^{(k+1)}$ with \eqref{policy improvement}.}
		\State \textbf{Convergence Check}
		\State \hspace{0.5cm}{If $\|u_{\varepsilon}^{(k+1)}-u_{\varepsilon}^{(k)}\|<\nu$, stop iteration and return the policy $u_{\varepsilon}^{(k+1)}$. Otherwise, set $k=k+1$ and continue the policy evaluation step.}
	\end{algorithmic}
	\label{alg1}
\end{algorithm}
\begin{remark}[Underlying design philosophy]
	To address the strong coupling between $V_{ro}^*$ and $u_{ro}^*$ in \eqref{robust HJI} and \eqref{robust optimal controller}, Algorithm~1 is built upon a key insight that, the impact of uncertainty $(c_1,c_2)$ is temporarily fixed with calculation results from the previous iteration. From a computational perspective, this renders the right-hand side of \eqref{policy evaluation} known, effectively reducing \eqref{robust HJI} to a standard HJB form (with regularization) that can be efficiently solved using existing numerical tools. From a theoretical viewpoint, since a well-identified bilinear lifting \eqref{lifted system AB} typically yields small error coefficients $(c_1,c_2)$, the associated terms act as mild perturbations. Consequently, fixing them with $V_\varepsilon^{(k)}$ and $u_\varepsilon^{(k)}$ does not significantly alter the solution at each iteration. This deliberate yet justifiable approximation transforms a challenging problem into a sequence of tractable subproblems, which is later shown convergent to the solution of \eqref{robust HJI}.
\end{remark}

The robust optimal control formulation involves normalization terms depending on $\|\nabla V\|$ and $\|u\|$, which may become ill-defined near the equilibrium point. To avoid potential singularities and ensure the well-posedness of proposed algorithm, we restrict our analysis to a punctured domain excluding a small neighborhood of the origin. Specifically, consider a bounded open set $\Omega\in\mathbb{R}^N$ containing the origin, and define a punctured domain $\Omega_\rho=\Omega\backslash B_\rho(0)$ where $\rho>0$ is a small constant and $B_\rho(0)$ denotes a ball of radius $\rho$. Further, we assume all of the considered control inputs $u$ belong to $\mathbb{R}^m \backslash B_{\rho'}(0)$. Since the origin is an equilibrium point of closed-loop system and thus the value function is normalized as $V(0)=0$, additional homogeneous Dirichlet boundary conditions are imposed to ensure the closed-loop stability, i.e., $V=0, \nabla V=0$ on $\partial B_\rho$.
\begin{remark}[Inside small neighborhood of origin]
	Since the proposed algorithm might become ill-defined near the equilibrium, we can implement the algorithm outside $B_\rho(0)$ with a sufficiently small $\rho>0$. As for $z\in B_\rho(0)$, one can find a state-feedback design $u=Kz$ linearly dependent on the lifted state with a primary objective to ensure the closed-loop robust stability \cite{Strasser2024Koopman,lin2025integratinguncertainties}. This thought has been also confirmed effective by existing works on Koopman-based optimal control \cite{HuangUmesh2022ConvexPFKO}.
\end{remark}
Now we proceed to establish the well-posedness and convergence of proposed algorithm. Define
\begin{equation}\label{definition of H delta}
	\begin{aligned}
		H(z,p)=p^\top Az+\frac{1}{2}z^\top Qz-\frac{1}{2}p^\top B(z)R^{-1}B^\top(z)p, \\
		\delta(z,u,p)=-(c_1\|z\|+c_2\|u\|)\|p\|-\frac{c_2^2\|p\|^2}{2\|u\|^2}u^\top R^{-1}u.
	\end{aligned}
\end{equation}
Then the policy evaluation step \eqref{policy evaluation} is written as
\begin{equation}\label{rewritten policy evaluation}
	-\varepsilon\nabla^2 V_{\varepsilon}^{(k+1)}+H(z,\nabla V_{\varepsilon}^{(k+1)}) =\delta(z,u_{ro}^{(k)},\nabla V_{\varepsilon}^{(k)})
\end{equation}
which is an elliptic second-order PDE. Meanwhile, equation \eqref{robust HJI} is equivalently written as
\begin{equation}\label{rewritten robust HJI}
	H(z,\nabla V_{ro}^*)=\delta(z,u_{ro}^*,\nabla V_{ro}^*).
\end{equation}
To connect \eqref{rewritten policy evaluation} and \eqref{rewritten robust HJI}, a natural intermediate is
\begin{equation}\label{rewritten intermediate}
	-\varepsilon\nabla^2 V_{\varepsilon}^*+H(z,\nabla V_{\varepsilon}^*)=\delta(z,u_{\varepsilon}^*, \nabla V_{\varepsilon}^*)
\end{equation}
where $u_{\varepsilon}^*$ is obtained with \eqref{robust optimal controller} substituting $V_{ro}^*$ with $V_{\varepsilon}^*$. We will first prove that the policy iteration algorithm promises the convergence $V_{\varepsilon}^{(k)}\rightarrow V_{\varepsilon}^*, u_{\varepsilon}^{(k)}\rightarrow u_{\varepsilon}^*$ for any fixed $\varepsilon>0$, and $V_{\varepsilon}^*\rightarrow V_{ro}^*, u_{\varepsilon}^*\rightarrow u_{ro}^*$ when $\varepsilon \rightarrow0$.

\begin{assumption}\label{Assum: initial bounded}
	The initial admissible policy and the gradient of corresponding value function are essentially bounded, i.e., $u_{\varepsilon}^{(0)}\in L^\infty(\Omega_\rho), \nabla V_{\varepsilon}^{(0)}\in L^\infty(\Omega_\rho)$.
\end{assumption}
\begin{lemma}\label{Lem: Properties of Vkuk}
	Let Assumption~\ref{Assum: initial bounded} holds. For any fixed $\varepsilon>0$, there exists a unique classical solution $V_{\varepsilon}^{(k+1)}\in \mathcal{C}^{2,\alpha} (\Omega_\rho)$ for \eqref{rewritten policy evaluation} with homogeneous boundary conditions at each iteration. Moreover, there exist $C_{p,1},C_{u,1}>0$ independent of $k$, such that for all iterations $\|\nabla V_{\varepsilon}^{(k)}\|_{ L^\infty(\Omega_\rho)}\leq C_{p,1}$ and $\|u_{\varepsilon}^{(k)}\|_{L^\infty(\Omega_\rho)}\leq C_{u,1}$ uniformly.
\end{lemma}
See Appendix~\ref{Proof: Properties of Vkuk} for the proof.
\begin{lemma}\label{Lem: Properties of V*u*}
	For any fixed $\varepsilon>0$, there exists a unique classical solution $V_{\varepsilon}^*\in \mathcal{C}^{2,\alpha}(\Omega_\rho)$ for \eqref{rewritten intermediate} with homogeneous boundary conditions, which admits the existence of constants $C_{p,2}, C_{u,2}>0$ such that $\|\nabla V_{\varepsilon}^*\|_{ L^\infty(\Omega_\rho)} \leq C_{p,2}$ and $\|u_{\varepsilon}^*\|_{L^\infty(\Omega_\rho)}\leq C_{u,2}$.
\end{lemma}
See Appendix~\ref{Proof: Properties of V*u*} for the proof.
Lemma~\ref{Lem: Properties of Vkuk} has illustrated the well-posedness of proposed policy iteration algorithm, and Lemma~\ref{Lem: Properties of V*u*} has demonstrated properties of our desired limit case. To establish the convergence, we need the following result which characterizes the continuous dependence of PDE solution on the variation of right hand side in \eqref{rewritten policy evaluation} and \eqref{rewritten intermediate}.
\begin{proposition}\label{Prop: Continuous Dependence}
	There exists a constant $C_\varepsilon>0$ such that
	\begin{equation}\label{Continuous Dependence}
		\|\nabla V_\varepsilon^{(k+1)}-\nabla V_\varepsilon^*\|_{L^\infty(\Omega_\rho)}\leq C_\varepsilon \|\delta^{(k)}-\delta^*\|_{L^\infty(\Omega_\rho)}.
	\end{equation}
	where $\delta^{(k)}=\delta(z,u_{ro}^{(k)},\nabla V_{\varepsilon}^{(k)})$ and $\delta^*=\delta(z,u_{ro}^*,\nabla V_{\varepsilon}^*)$.
\end{proposition}
See Appendix~\ref{Proof: Continuous Dependence} for the proof. Building on the preceding results, we can now establish the convergence of Algorithm~1.
\begin{theorem}\label{Thm: Convergence}
	Let Assumption~\ref{Assum: initial bounded} holds. For any fixed $\varepsilon>0$ and sufficiently small error bound coefficients $c_1,c_2>0$, the sequences of value functions $\{V_{\varepsilon}^{(k)}\}_{k=0}^\infty$ and corresponding controllers $\{u_{\varepsilon}^{(k)}\}_{k=0}^\infty$ calculated by \eqref{policy evaluation} and \eqref{policy improvement} converge to $V_{\varepsilon}^*$ and $u_{\varepsilon}^*$ respectively. Moreover, as $\varepsilon\rightarrow0$, $V_{\varepsilon}^*\rightarrow V_{ro}^*$ where $V_{ro}^*$ is the unique viscosity solution of \eqref{rewritten robust HJI} (equivalently \eqref{robust HJI}), and therefore $u_{\varepsilon}^*\rightarrow u_{ro}^*$.
\end{theorem}
\begin{pf}
	With $\|\nabla V_{ro}^{(k+1)}\|-\|\nabla V_{ro}^{(k)}\|\leq\|\nabla V_{ro}^{(k+1)}-\nabla V_{ro}^{(k)}\|$ and the definition of $\delta$ by \eqref{definition of H delta}, we can prove that
	\begin{equation*}
		\begin{aligned}
			\|\delta^{(k)}-\delta^*\|_{L^\infty(\Omega_\rho)}&\leq \overline{L}_1\|\nabla V_{\varepsilon}^{(k)}-\nabla V_{\varepsilon}^*\|_{L^\infty(\Omega_\rho)} \\
			&+\overline{L}_2 \|u_{\varepsilon}^{(k)}-u_{\varepsilon}^*\|_{L^\infty(\Omega_\rho)}
		\end{aligned}
	\end{equation*}
	where
	\begin{equation*}
		\begin{aligned}
			\overline{L}_1&=\left(c_1\|z\|+c_2\|u_{\varepsilon}^{(k)}\|+\frac{c_2^2(\|\nabla V_{\varepsilon}^{(k)}\|+\|\nabla V_{\varepsilon}^*\|)}{2\lambda_{\min}(R)} \right)_{L^\infty(\Omega_\rho)} \\ &\leq c_1\|z\|_{L^\infty(\Omega_\rho)}+c_2C_{u,1} +\frac{c_2^2(C_{p,1}+C_{p,2})}{2\lambda_{\min}(R)}=L_1.
		\end{aligned}
	\end{equation*}
	and
	\begin{equation*}
		\begin{aligned}
			\overline{L}_2&=\left(c_2\|\nabla V_{\varepsilon}^*\|+\frac{c_2^2\|\nabla V_{\varepsilon}^*\|^2L_{u,1}}{2}\right)_{L^\infty(\Omega_\rho)} \\
			&\leq c_2C_{u,2}\left(1+\frac{c_2 C_{u,2}L_{u,1}}{2}\right)=L_2.
		\end{aligned}
	\end{equation*}
	Here $L_{u,1}$ is the Lipschitz constant of $u^\top R^{-1}u/\|u\|^2$ since $u\in\mathbb{R}^m\backslash B_{\rho'}(0)$ and we have proved the boundedness of $u_{\varepsilon}^{(k)}$ and $u_{\varepsilon}^*$. In addition, the inequality
	\begin{equation*}
		u^\top R^{-1}u\leq \lambda_{\max}(R^{-1})\|u\|^2=\frac{1}{\lambda_{\min}(R)}\|u\|^2
	\end{equation*}
	is used to obtain $L_1,L_2$. With Proposition~\ref{Prop: Continuous Dependence}, we obtain
	\begin{equation}\label{error inequality 1}
		\begin{aligned}
			\|\nabla V_{\varepsilon}^{(k+1)}-\nabla V_{\varepsilon}^*\|_{L^\infty(\Omega_\rho)} \leq C_\varepsilon L_1\|\nabla V_{\varepsilon}^{(k)}-\nabla V_{\varepsilon}^*\|_{L^\infty(\Omega_\rho)} \\ +C_\varepsilon L_2 \|u_{\varepsilon}^{(k)}-u_{\varepsilon}^*\|_{L^\infty(\Omega_\rho)}.
		\end{aligned}
	\end{equation}
	Combining \eqref{policy improvement} with \eqref{robust optimal controller}, we obtain
	\begin{equation*}
		\begin{aligned}
			u_{\varepsilon}^{(k+1)}&-u_{\varepsilon}^*=-R^{-1}B^\top(z)\left(\nabla V_{\varepsilon}^{(k+1)}-\nabla V_{\varepsilon}^*\right) \\
			&-c_2R^{-1}\left(\|\nabla V_{\varepsilon}^{(k+1)} \|\frac{u_{\varepsilon}^{(k)}}{\|u_{\varepsilon}^{(k)}\|}-\|\nabla V_{\varepsilon}^*\|\frac{u_{\varepsilon}^*}{\|u_{\varepsilon}^*\|}\right).
		\end{aligned}
	\end{equation*}
	Denote $L_{u,2}$ as the Lipschitz constant of $u/\|u\|$, then
	\begin{equation}\label{error inequality 2}
		\begin{aligned}
			\|u_{\varepsilon}^{(k+1)}-u_{\varepsilon}^*\|_{L^\infty(\Omega_\rho)}\leq L_4\|u_{\varepsilon}^{(k)}-u_{\varepsilon}^*\|_{L^\infty(\Omega_\rho)} \\
			+L_3\|\nabla V_{\varepsilon}^{(k+1)}-\nabla V_{\varepsilon}^*\|_{L^\infty(\Omega_\rho)}
		\end{aligned}
	\end{equation}
	where
	\begin{equation*}
		L_3=\|R^{-1}\|(c_2+\|B(z)\|_{L^\infty(\Omega_\rho)}), \ L_4=c_2C_{p,1}L_{u,2}\|R^{-1}\|.
	\end{equation*}
	Define $e_{k}=[\|\nabla V_{\varepsilon}^{(k)}-\nabla V_{\varepsilon}^*\|_{ L^\infty(\Omega_\rho)} \quad \|u_{\varepsilon}^{(k)}-u_{\varepsilon}^*\|_{L^\infty(\Omega_\rho)} ]^\top$. Then \eqref{error inequality 1} and \eqref{error inequality 2} can be written as
	\begin{equation*}
		\begin{bmatrix}
			1& 0 \\ -L_3& 1
		\end{bmatrix}e_{k+1}\leq\begin{bmatrix}
			C_\varepsilon L_1& C_\varepsilon L_2 \\ 0& L_4
		\end{bmatrix}e_{k},
	\end{equation*}
	which is equivalent as
	\begin{equation*}
		e_{k+1}\leq\begin{bmatrix}
			C_\varepsilon L_1& C_\varepsilon L_2 \\ C_\varepsilon L_1L_3& C_\varepsilon L_2L_3+L_4
		\end{bmatrix}e_{k}\triangleq Ee_{k}.
	\end{equation*}
	Note that all elements of $E$ are $\mathcal{O}(c_1+c_2)$, which means that they can be sufficiently small with appropriate $c_1,c_2$. When $\|E\|<1$ (or $|\lambda(E)|<1$), $\lim_{k\rightarrow\infty}e_{k}=0$, which further indicates that $V_{\varepsilon}^{(k)} \rightarrow V_{\varepsilon}^*$ and $u_{\varepsilon}^{(k)}\rightarrow u_{\varepsilon}^*$ in the sense of $L^\infty$ norm. Finally, the convergence $\lim_{\varepsilon\rightarrow0}V_\varepsilon^*=V_{ro}^*$ follows from the vanishing viscosity theory \cite[Section 6]{crandall1992users}\cite{Barles2013}, which completes the proof. \qed
\end{pf}

\section{Numerical Examples}\label{6.numerical}
In this section, numerical simulations using MATLAB are conducted to verify the theoretical results developed in the previous sections and to demonstrate the effectiveness of the proposed robust optimal control methodology as well as corresponding policy iteration algorithm.

Consider the optimal control problem of the control-affine nonlinear system
\begin{equation}\label{simulation original system}
	\begin{bmatrix}
		\dot{x}_1 \\ \dot{x}_2
	\end{bmatrix}=\begin{bmatrix}
		-x_1+x_2 \\ -\frac{1}{2}(x_1+x_2)+\frac{1}{2}x_1^2x_2
	\end{bmatrix}+\begin{bmatrix}
		0 \\ x_1
	\end{bmatrix}u
\end{equation}
where the weighting matrices of quadratic index \eqref{optimal control problem} are set as $\overline{Q}=I_{2\times2},R=1$. The analytical optimal value function and corresponding optimal controller are verifiable as $V^*(x)=\frac{1}{4}x_1^2+\frac{1}{2}x_2^2$ and $u^*(x)=-x_1x_2$, which serve as the ground truth for performance comparison.

The dynamics \eqref{simulation original system} is only used for collecting time-series data $\left\{x_j,\dot{x}_j\right\}_{j=1}^{T}$ with total amount $T=5000$ corrupted by the bounded noise $\overline{d}(t)=0.01\cdot[\cos(2\pi 0.4t)\ \sin(2\pi 0.4t)]$. The dictionary functions are constructed as monomials up to order $3$ of $x$, i.e., $\Psi(x)=[x_1\ x_2\ \frac{1}{2}x_1^2\ x_1x_2\ \frac{1}{2}x_2^2\ \ldots]^\top$. The lifted bilinear system of dimension $N=9$ is identified from the collected data via EDMD algorithm \eqref{identified Z}. Meanwhile, denote
\begin{equation*}
	R=[r_1\ r_2\ \cdots \ r_{T}]=Z_1-\left[A\ B_0\ B_1\ \cdots\ B_m\right]W_0
\end{equation*}
which records the approximation error at all data points, and
\begin{equation*}
	\begin{aligned}
		\tilde{R}=[\|r_1\| \ \|r_2\| \ \cdots \ \|r_{T}\|], 
		\tilde{Z}_0=[\|z_1\| \ \|z_2\| \ \cdots \ \|z_{T}\|],\\
		\tilde{U}_0=[\|u_1\| \ \|u_2\| \ \cdots \ \|u_{T}\|].
	\end{aligned}
\end{equation*}
The coefficients of approximation error bound \eqref{error bound} can then be solved from data via linear programming under the constraint
\begin{equation}
	\tilde{R}\leq c_1\tilde{Z}_0+c_2\tilde{U}_0,
\end{equation}
and we obtain the error bound coefficients $c_1=c_2=0.1435$.

To numerically compute the solution of policy evaluation PDE \eqref{policy evaluation}, the Galerkin method is adopted, which is widely used for solving nonlinear PDE by projecting the original equation onto a finite-dimensional function space \cite{bertsekas1995dynamic}. Specifically, the value function is approximated using a set of linearly independent basis functions $\{\varphi_i(z)\}_{i=1}^M$, i.e.,
\begin{equation}\label{Galerkin approximation}
	V_\varepsilon^{(k+1)}(z) \approx \sum_{i=1}^{M} \theta^{(k+1)}_i \varphi_i(z)\triangleq {\theta^{(k+1)}}^\top \varphi(z).
\end{equation}
Substituting \eqref{Galerkin approximation} into the policy evaluation PDE \eqref{policy evaluation} yields a residual, which is then minimized in a least-squares sense over a set of $N_c=5000$ randomly sampled collocation points. In this way, solving the nonlinear PDE \eqref{policy evaluation} is reduced to a tractable algebraic problem with regard to $\theta\in\mathbb{R}^M$ and consequently, Algorithm 1 can be effectively implemented. In our simulation, the basis functions $\{\varphi_i(z)\}_{i=1}^M$ are chosen as monomials of $z$ with order $2$. Since the lifted coordinates $z$ already contain monomials of the original state variable $x$, some candidate basis functions become linearly dependent. Such dependencies are removed from $\varphi(z)$ to avoid the algebraic singularity in the least-square regression and ensure the convergence of policy iteration, reducing the number of basis function to $M=25$. The diffusion parameter for elliptic regularization in \eqref{policy evaluation} is set as $\varepsilon=10^{-3}$.

\begin{figure}[htbp]
	\centering
	\includegraphics[width=0.45\textwidth]{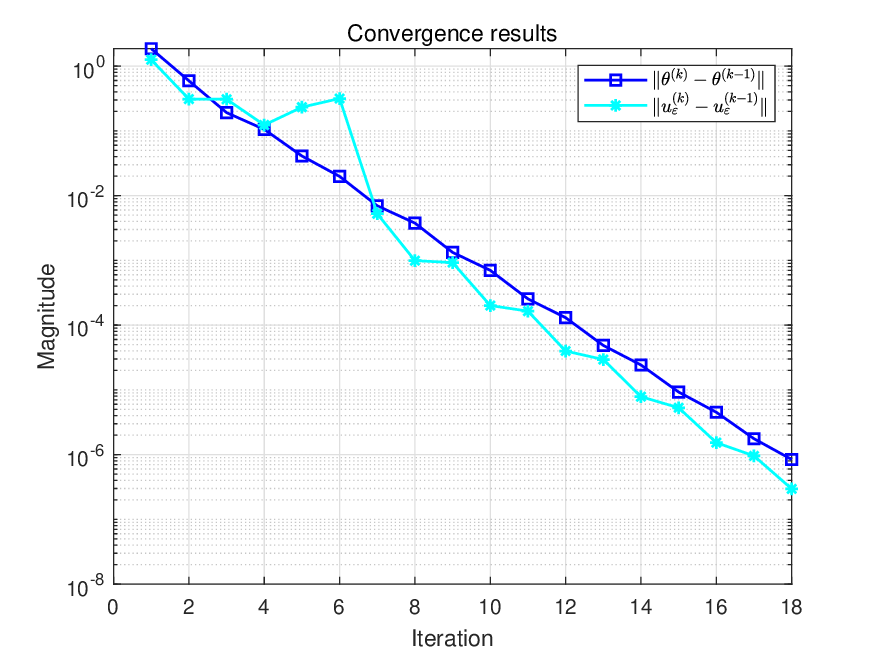}
	\caption{The relative changes of value function coefficients and the control policy at each iteration step, which shows the convergence of Algorithm~1. The entire computation finishes in less than one second.}
	\label{fig:convergence}
\end{figure}
\begin{figure}[htbp]
	\centering
	\includegraphics[width=0.48\textwidth]{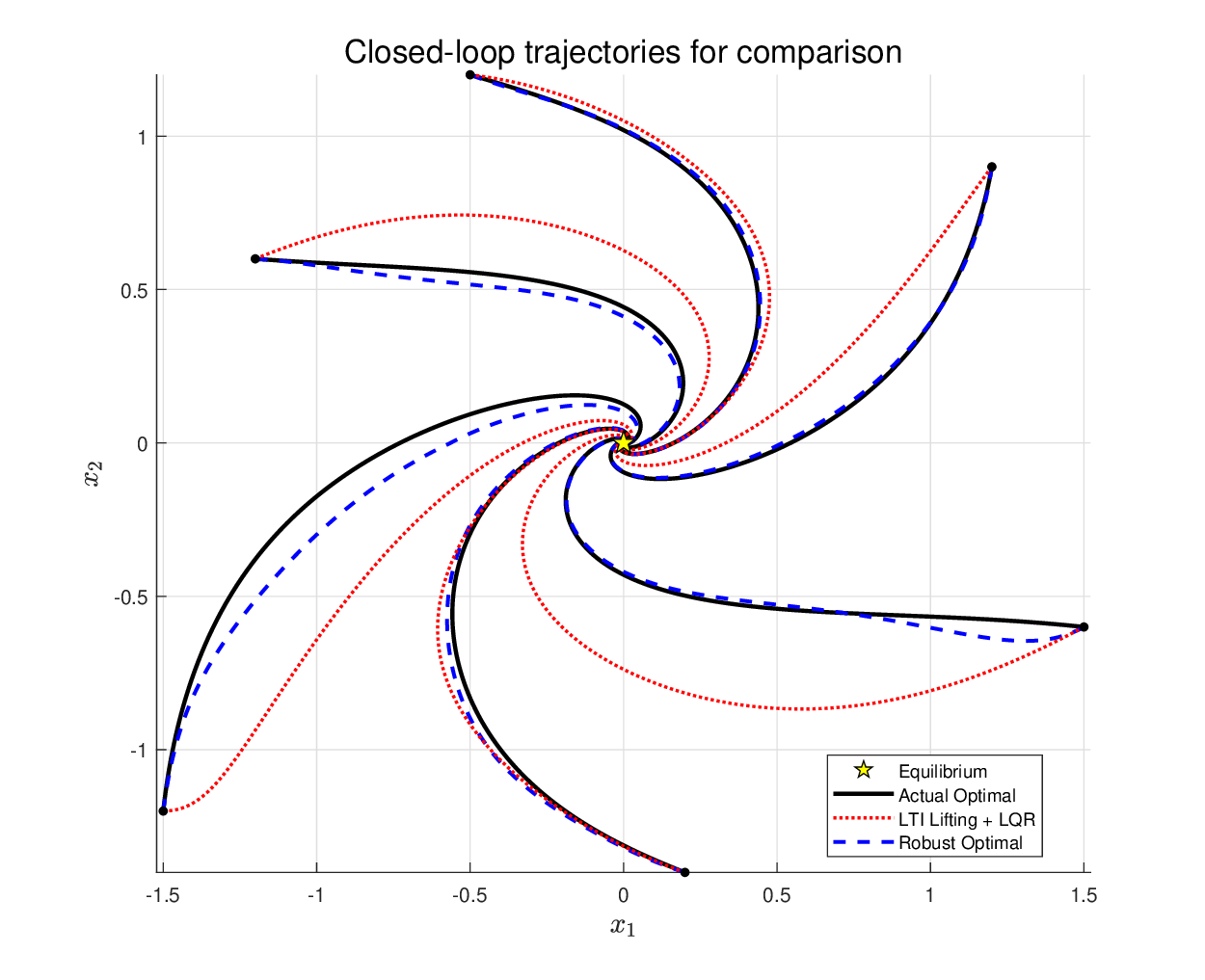}
	\caption{Comparison of closed-loop trajectories starting from 6 randomly chosen initial points. Three control strategies are evaluated, i.e., the actual optimal controller $u^*$ (black solid), standard LTI lifting then LQR (red dotted), and the robust optimal controller $u_{ro}^*(u_\varepsilon^*)$. All trajectories illustrate the convergence behavior toward the equilibrium.}
	\label{fig:trajectory}
\end{figure}

Starting from an initial policy, coefficients $\theta^{(k)}$ of the value function and the control policy $u_\varepsilon^{(k)}$ are iteratively updated according to the proposed policy iteration algorithm. As shown in Fig.~\ref{fig:convergence}, the solutions converge rapidly within 20 iterations, during which the variations of coefficients and control policy decrease significantly.

As illustrated in Fig.~\ref{fig:trajectory}, closed-loop trajectories controlled by the resulting robust optimal controller $u_{ro}^*(u_\varepsilon^*)$ closely track the actual optimal paths. Meanwhile, closed-loop trajectories generated by the LQR optimal controller based on the extensively utilized LTI lifted models are given for comparison. As discussed in Remark~\ref{Rem: LTI lifting}, the underlying severe modeling error in LTI lifting leads to substantial optimality deviations. Quantitatively, the performance costs of different control strategies are summarized in Table~\ref{table_performance_comparison}, which shows that our methodology and corresponding algorithm successfully mitigate the optimality deviations and recover near-optimal performance.

\begin{table}[!t]
	\renewcommand{\arraystretch}{1.3} 
	\centering
	\caption{Comparison of Performance Costs}
	\label{table_performance_comparison}
	\begin{tabular}{@{}
			p{1.8cm} 
			>{\centering\arraybackslash}p{0.8cm} 
			>{\centering\arraybackslash}p{0.8cm} 
			>{\centering\arraybackslash}p{0.8cm} 
			>{\centering\arraybackslash}p{1cm} 
			>{\centering\arraybackslash}p{1cm} @{} 
		}
		\toprule
		\multirow{2}{*}{\textbf{Initial State}} 
		& \multicolumn{3}{c}{\textbf{Cumulative Cost} $J$} 
		& \multicolumn{2}{c}{\textbf{Relative Extra}} \\
		\cmidrule(r){2-4} \cmidrule(l){5-6}
		& Actual & Robust & LQR & Robust & LQR \\
		\midrule
		$(-1.5,-1.2)$          & 1.2999 & 1.3494 & 3.0609 & 3.81\%  & 138.67\% \\
		$(-0.5,\phantom{-}1.2)$& 0.7877 & 0.8032 & 0.8487 & 3.91\%  & 7.74\%  \\
		$(\phantom{-}1.5,-0.6)$& 0.7511 & 0.7735 & 0.9346 & 2.98\%  & 24.43\% \\
		$(\phantom{-}1.2,\phantom{-}0.9)$& 0.7736 & 0.7776 & 1.0568 & 0.05\% & 36.61\% \\
		$(\phantom{-}0.2,-1.4)$& 0.9952 & 1.0454 & 1.1131 & 5.04\%  & 11.85\% \\
		$(-1.2,\phantom{-}0.6)$& 0.5458 & 0.5528 & 0.6371 & 1.28\%  & 16.73\% \\
		\midrule
		\textbf{Average} & -- & -- & -- & \textbf{2.85\%} & \textbf{39.34\%} \\
		\bottomrule
	\end{tabular}
\end{table}
In Remark~\ref{Rem: Tradeoff}, we have demonstrated a trade-off between the nominal performance and optimality robustness inherent in the proposed robust optimal control methodology. As visualized in Fig.~\ref{fig:Tradeoff}, the left panel indicates a relatively minor nominal performance loss, while the right one reveals a significantly larger optimality robustness improvement achieved by $u_{ro}^*$. In addition, a comparison of optimality deviations associated with $u_{ro}^*$ and $u_0^*$ is demonstrated in Fig.~\ref{fig:deviation comparison}. It is observed that the optimality deviation from the true optimal controller $u^*$ is effectively mitigated by our robust optimal control methodology, demonstrating an improved optimality robustness.

\begin{figure}[htbp]
	\centering
	\includegraphics[width=0.52\textwidth]{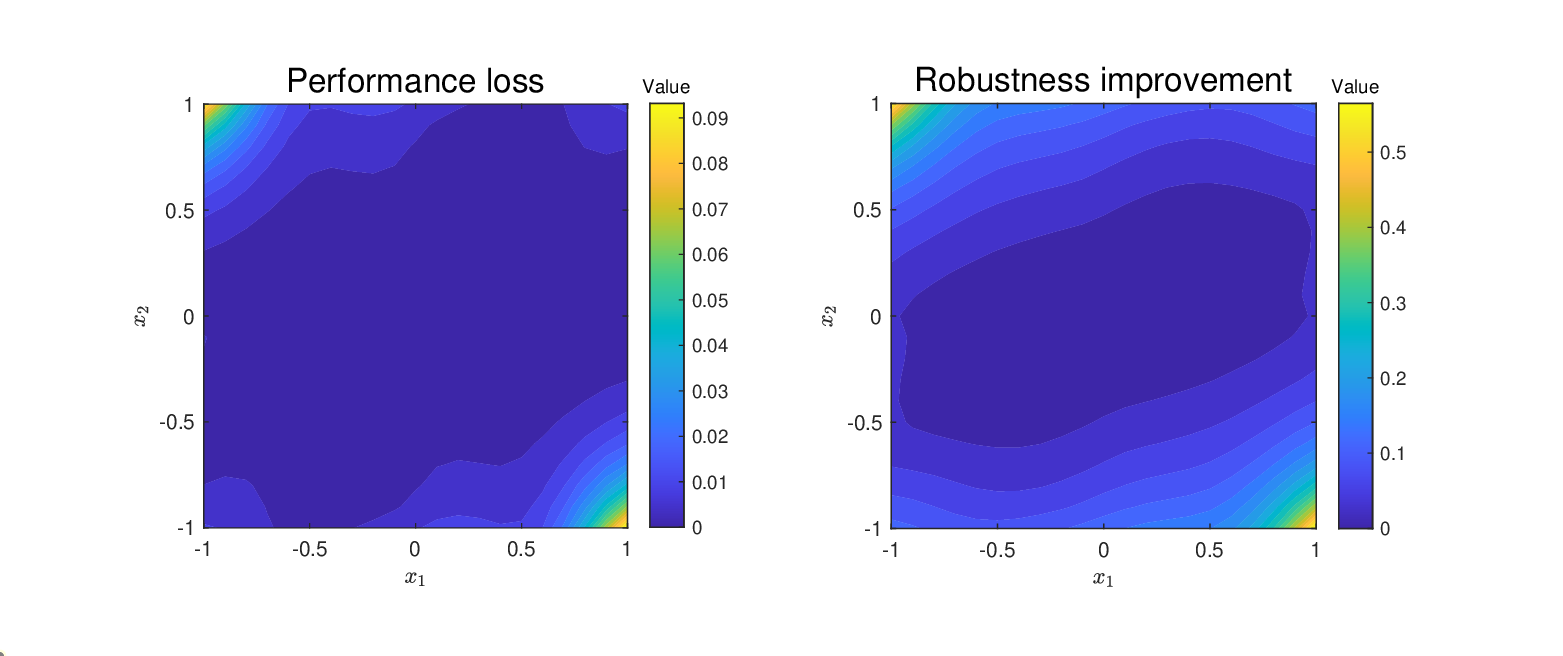}
	\caption{Quantification of performance-robustness trade-off described in Theorem~\ref{Thm: Robust Performance} and explained in Remark~\ref{Rem: Tradeoff}. The left panel displays the nominal performance loss of $u_{ro}^*$, given by \eqref{extra cost uro}. The right panel illustrates the potential optimality degradation of $u_0^*$ as well as the robustness improvement of $u_{ro}^*$, given by \eqref{performance improvement uro}.}
	\label{fig:Tradeoff}
\end{figure}
\begin{figure}[htbp]
	\centering
	\includegraphics[width=0.52\textwidth]{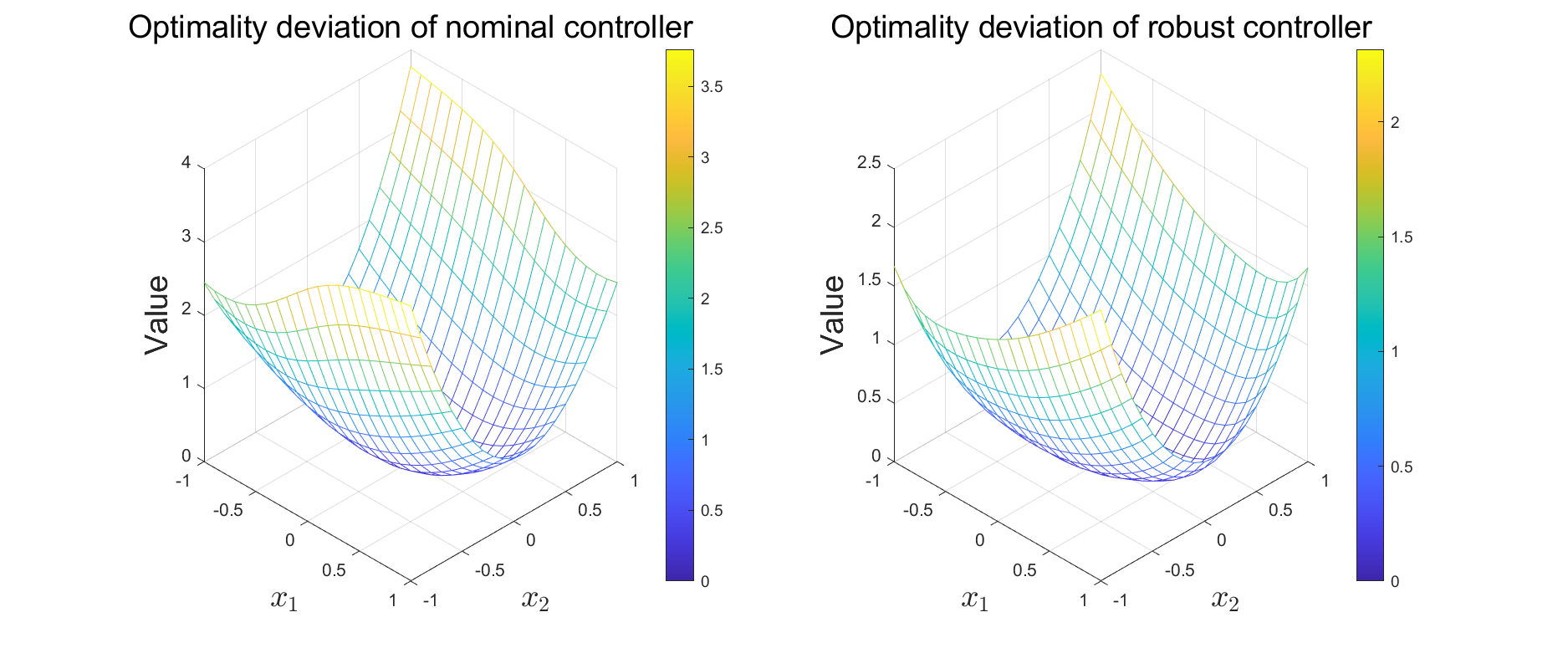}
	\caption{Comparison of the optimality deviations of the robust optimal controller $u_{ro}^*$ ($u_\varepsilon^*$, left) and the nominal one $u_0^*$ (right), respectively characterized in Theorems~2 and 5.}
	\label{fig:deviation comparison}
\end{figure}

In the present implementation, the Galerkin method is employed to compute the solution of PDE \eqref{policy evaluation}. While it provides an effective realization, the integration of more sophisticated numerical techniques may further improve computational efficiency and accuracy, which will be explored in future.


\section{Conclusions}\label{7.conclusion}
This paper has systematically studied optimality robustness in Koopman-based data-driven control subject to multi-source uncertainties. By developing a unified analytical framework, we have characterized how heterogeneous uncertainties affect optimal control performance and established principled mechanisms for mitigating the resulting optimality deviations. Our results provide a systematic analysis-to-design perspective for Koopman-based control that complements existing stability-oriented robustness theories with explicit optimality-oriented guarantees.

Beyond the methodology developed in this work, the underlying framework offers a general foundation for investigating optimality robustness in data-driven control. Future research directions include extending the present analysis to more general problem formulations or uncertainty structures, integrating adaptive and learning-based mechanisms for computation, and exploring applications in networked systems. In addition, the introduced analytical insights and techniques may facilitate the study of optimality robustness in broader classes of learning-enabled and model-based control architectures.
\begin{ack}                               
This work was financially supported by the National Natural Science Foundation of China (NSFC) under grants T2121002, U24A20266, and 62173006.  
\end{ack}

\bibliographystyle{plain}        
\bibliography{Automatica_References}           

\appendix
\section{Proof of Proposition~1}\label{Proof: Set Reformulation}
The full row rank of $W_0$ ensures $\mathbf{A}=W_0W_0^\top\succ 0$, whose pseudo-inverse satisfies $W_0^\dagger=W_0^\top (W_0W_0^\top)^{-1}$. For any element in $\mathcal{C}_0$, $\tilde{\mathbf{Z}}^\top=[A\ B_0\ B_1\ \ldots\ \ B_m]=(Z_1)W_0^\dagger=\boldsymbol{\zeta}^\top -D_0W_0^\dagger$, we have the following relation
\begin{equation}\label{C0 in C}
	(\tilde{\mathbf{Z}}-\boldsymbol{\zeta})^\top \mathbf{A} (\tilde{\mathbf{Z}}-\boldsymbol{\zeta})=D_0 W_0^\top(W_0W_0^\top)^{-1}W_0D_0^\top.
\end{equation}
Define $\mathbf{Q}_p=W_0^\top(W_0W_0^\top)^{-1}W_0\in\mathbb{R}^{T\times T}$, which is actually a projection matrix. We can prove that eigenvalues of $\mathbf{Q}_p$ are $0$ and $1$. Specifically, for $\lambda_p=0$, any non-zero vector $v\in\mathrm{Nul}(W_0)$ is an eigenvector of $\mathbf{Q}_p$. For $\lambda_p=1$, any non-zero vector $v\in\mathrm{Col}(W_0^\top)$ is an eigenvector of $\mathbf{Q}_p$, since any $0\neq v=W_0^\top w\in\mathrm{Col}(W_0^\top)$ allows $Q_pv=W_0^\top w=v$. Meanwhile, $\mathbb{R}^T=\mathrm{Col}(W_0^\top)\oplus\mathrm{Nul}(W_0)$ means that $\mathbf{Q}$ has no other eigenvalue except $0$ and $1$. Then $\mathbf{Q}_p$ can be bounded with identity matrix, i.e., $Q_p\preceq I$. With \eqref{C0 in C}, we obtain
\begin{equation*}
	(\tilde{\mathbf{Z}}-\boldsymbol{\zeta})^\top \mathbf{A} (\tilde{\mathbf{Z}} -\boldsymbol{\zeta})\preceq D_0D_0^\top\preceq\Delta\Delta^\top=\mathbf{Q},
\end{equation*}
i.e., any $\tilde{\mathbf{Z}}^\top\in\mathcal{Z}_0$ belongs to $\mathcal{Z}$. The equivalent relation $\mathcal{Z}=\mathcal{E}$ follows \cite[Proposition 1]{BISOFFI2022Petersen}. The proof is completed.

\section{Proof of Lemma~1}\label{Proof: Properties of Vkuk}
The proof is derived by mathematical induction. For $k=0$, Assumption~\ref{Assum: initial bounded} implies that $\delta(z,u_\varepsilon^{(0)},\nabla V_{\varepsilon}^{(0)})\in L^\infty$. Further, the definition \eqref{definition of H delta} admits that $H(z,p)$ can be bounded with $\|p\|$. As the elliptic PDE \eqref{rewritten policy evaluation} is quasilinear after introducing the vanishing viscosity regularization $\varepsilon \nabla^2$, the elliptic PDE theory \cite[Theorem 11.4]{gilbarg2001elliptic} ensures the existence of a unique classical solution $V_\varepsilon^{(1)}\in\mathcal{C}^{2,\alpha}(\Omega_\rho)$, and $u^{(1)}\in L^\infty$ holds with \eqref{policy improvement}. Suppose $u_{\varepsilon}^{(k)}\in L^\infty$ and $\nabla V_{\varepsilon}^{(k)}\in L^\infty$, then $\delta(z,u_\varepsilon^{(k)},\nabla V_{\varepsilon}^{(k)})\in L^\infty$. Hence $u_{\varepsilon}^{(k+1)}\in L^\infty$ is similarly ensured together with the uniqueness of classical solution $V_{\varepsilon}^{(k+1)}\in \mathcal{C}^{2,\alpha}(\Omega_\rho)\subseteq L^\infty (\Omega_\rho)$ and \eqref{policy improvement}. Consequently, the uniform boundedness naturally holds.

\section{Proof of Lemma~2}\label{Proof: Properties of V*u*}
We rewrite \eqref{rewritten intermediate} as
\begin{equation*}
	-\varepsilon\nabla^2 V_{\varepsilon}^*+H(z,\nabla V_{\varepsilon}^*)-\delta(z, u_{\varepsilon}^*,\nabla V_{\varepsilon}^*)=0
\end{equation*}
where $H(z,p)-\delta(z,u,p)$ is continuously differentiable as $\|u\|\geq\rho'$. Then the uniqueness of classical solution $V_{\varepsilon}^*\in C^{2,\alpha}(\Omega_\rho)$ is ensured by the elliptic PDE theory \cite[Theorem 10.2]{gilbarg2001elliptic}. A key observation on the robust optimal controller given by \eqref{robust optimal controller} is that $u_{\varepsilon}^*$ is Lipschitz continuous on $\nabla V_{\varepsilon}^*$ since we assume $u\in\mathbb{R}^m\backslash B_\rho'(0)$. Hence, there exists $C_{u,2}>0$ such that $\|u_{\varepsilon}^*\|_{L^\infty(\Omega_\rho)}\leq C_{u,2}$.

\section{Proof of Proposition~2}\label{Proof: Continuous Dependence}
Denote $W^{(k+1)}=V_\varepsilon^{(k+1)}-V_\varepsilon^*$. Combining \eqref{rewritten policy evaluation} and \eqref{rewritten intermediate}, $W^{(k+1)}$ satisfies
\begin{equation*}
	-\varepsilon\nabla^2 W^{(k+1)} + b^\top(z)\cdot\nabla W^{(k+1)}=\delta^{(k)}-\delta^*.
\end{equation*}
where
\begin{equation*}
	b(z)=\int_0^1 \frac{\partial H}{\partial p}\left(\nabla V_\varepsilon^*+s\left(\nabla V_\varepsilon^{(k+1)}-\nabla V_\varepsilon^*\right)\right)ds
\end{equation*}
is obtained by mean value theorem for multivariate functions. By Lemmas~\ref{Lem: Properties of Vkuk} and \ref{Lem: Properties of V*u*} together with \eqref{definition of H delta}, $b(z)$ is bounded. The dependence of $\nabla W^{(k+1)}$ on $\delta^{(k)}-\delta^*$ follows from the maximum principle and elliptic estimates \cite[Lemma 9.17]{gilbarg2001elliptic}, which completes the proof.

\end{document}